\documentclass
  [
    DIV=14,
    fontsize=12,
    abstract=true,
    bibliography=totoc
  ]
  {scrartcl}

\setkomafont{date}{\small\rmfamily}
\setkomafont{author}{\large\sffamily}

\usepackage{authblk}

\newcommand
{\NameEmail}[2]
{%
  \stackunder%
  {#1}%
  {%
   \clap{%
   \normalfont%
   \footnotesize%
   \texttt{#2}%
  }}%
  \hspace{-3.5pt}%
}

\newcommand\keywords[1]{%
  \begingroup
  \renewcommand\thefootnote{}
  \footnote{{Keywords:} #1}%
  \addtocounter{footnote}{-1}
  \endgroup
}
\newcommand\MSCcodes[1]{%
  \begingroup
  \renewcommand\thefootnote{}
  \footnote{{MSC codes:} #1}%
  \addtocounter{footnote}{-1}
  \endgroup
}
\newcommand\funding[1]{%
  \begingroup
  \renewcommand\thefootnote{}
  \footnote{{Funding:} #1}%
  \addtocounter{footnote}{-1}
  \endgroup
}

\usepackage{scrlayer-scrpage}
\clearpairofpagestyles 
\ihead{} 
\ohead{\headmark}
\automark{section}

\NewCommandCopy{\oldparagraph}{\paragraph}
\RenewDocumentCommand{\paragraph}{s t{!} O{#4} m}{%
  \IfBooleanTF{#1}
    {\IfBooleanTF{#2}{\oldparagraph*{#4}}{\oldparagraph*{#4.}}}
    {\IfBooleanTF{#2}{\oldparagraph[#3]{#4}}{\oldparagraph[#3]{#4.}}}%
}

\NewCommandCopy{\oldsubparagraph}{\subparagraph}
\RenewDocumentCommand{\subparagraph}{s t{!} O{#4} m}{%
  \IfBooleanTF{#1}
    {\IfBooleanTF{#2}{\oldsubparagraph*{#4}}{\oldsubparagraph*{#4.}}}
    {\IfBooleanTF{#2}{\oldsubparagraph[#3]{#4}}{\oldsubparagraph[#3]{#4.}}}%
}

\usepackage{caption}
\usepackage
  [leftcaption]
  {sidecap}

\usepackage{enumitem}
\setlist[enumerate,1]{label=\textbf{\textup{(\roman*)}}}
\setlist[enumerate,2]{label=\textbf{\textup{(\alph*)}}}
\setlist{
  topsep=1pt,
  itemsep=2pt,
  leftmargin=1cm,
  before=\leavevmode
}

\usepackage[
  backend=biber,
  style=alphabetic,
  giveninits=true, 
  backref=true
]{biblatex}
\usepackage[T1]{fontenc}
\usepackage{lmodern}
\usepackage{anyfontsize}

\usepackage{multirow}
\usepackage{hhline}

\usepackage{tabularray}
\UseTblrLibrary{booktabs}

\usepackage{adjustbox}
\usepackage{tcolorbox}

\newtcolorbox{standout}{
  colback=gray!15,
  boxrule=0pt,
  left=.3cm,
  right=.3cm,
  top=.18cm,
  bottom=.18cm,
  boxsep=0pt
}

\usepackage{amsmath}
\allowdisplaybreaks
\usepackage{mathtools}
\usepackage{amssymb}
\usepackage{amsthm}
\usepackage{xfrac}
\usepackage{stmaryrd} 
\usepackage{scalerel} 
\usepackage{stackengine} 
\usepackage[safe]{tipa} 
\usepackage{centernot}
\usepackage[
  cal=euler,
  scr=dutchcal
]
 {mathalpha} 

 \newcommand{\bracket}[3]{%
  \stretchleftright
    {#1}
    {%
      \ensurestackMath{\addstackgap[1pt]{#2}}%
      \vrule width 0pt depth 2pt height 0pt
    }
    {#3}%
} 

\newcommand{\bracketmid}[4]{%
  \stretchleftright{#1}{%
    \ensurestackMath{%
      \addstackgap[1pt]{#2}%
      \,\stretchrel*{|}{\addstackgap[1pt]{#2#3}}\,%
      \addstackgap[1pt]{#3}%
    }%
  }{#4}%
}

\newcommand{\bracketmidtwo}[5]{%
  \stretchleftright{#1}{%
    \ensurestackMath{%
      \addstackgap[1pt]{#2}%
      \,\stretchrel*{|}{\addstackgap[1pt]{#2#3#4}}\,%
      \addstackgap[1pt]{#3}%
      \,\stretchrel*{|}{\addstackgap[1pt]{#2#3#4}}\,%
      \addstackgap[1pt]{#4}%
    }%
  }{#5}%
}

\theoremstyle{plain}
\newtheorem{theorem}{Theorem}[section]
\newtheorem{lemma}[theorem]{Lemma}
\newtheorem{proposition}[theorem]{Proposition}

\theoremstyle{definition}
\newtheorem{definition}[theorem]{Definition}
\newtheorem{example}[theorem]{Example}

\theoremstyle{remark}
\newtheorem{remark}[theorem]{Remark}
\usepackage[
  colorlinks=true,
  linkcolor=darkgreen,
  citecolor=darkgreen,
  urlcolor=darkgreen
]{hyperref}
\usepackage{xurl} 

\DeclareSourcemap{
  \maps[datatype=bibtex]{
    \map{
      \step[fieldsource=url, final=true]
      \step[fieldset=verba, origfieldval, final=true]
      \step[fieldsource=verba, match=\regexp{\Ahttps?://}, replace={}]
    }
  }
}
\DeclareFieldFormat{url}{%
  \mkbibacro{URL}\addcolon\space
  \href{#1}{\nolinkurl{\thefield{verba}}}%
}

\usepackage{cleveref}
\crefname{equation}{}{}
\crefname{section}{\S}{\S\S}
\crefname{subsection}{\S}{\S\S}
\crefname{subsubsection}{\S}{\S\S}
\crefname{definition}{Def.}{Defs.}
\crefname{theorem}{Thm.}{Thms.}
\crefname{corollary}{Cor.}{Cors.}
\crefname{lemma}{Lem.}{Lems.}
\crefname{proposition}{Prop.}{Props.}
\crefname{remark}{Rem.}{Rems.}
\crefname{notation}{Ntn.}{Ntns.}
\crefname{fact}{Fact}{Facts}
\crefname{example}{Ex.}{Exs.}
\crefname{figure}{Fig.}{Figs.}
\crefname{table}{Tab.}{Tabs.}
\crefname{footnote}{ftn.}{ftns.}
\Crefname{footnote}{Ftn.}{Ftns.}

\usepackage{tikz}
\usetikzlibrary{
  cd,
  calc,
  arrows.meta,
  backgrounds,
  decorations,
  decorations.pathmorphing, 
}

\tikzcdset{
  arrow style=tikz, 
  diagrams={
    cramped,
    baseline=
    {([yshift=-2pt]current bounding box.center)},
    >={
      Computer Modern Rightarrow[
        length=4pt, width=4pt
      ]
    },
    hook/.style={
      {Hooks[right, length=2pt, width=8pt]}->
    },
    hook'/.style={
      {Hooks[left, length=2pt, width=8pt]}->
    },
    shorten=0pt,
  }
}

\usepackage{xcolor}
\definecolor{darkblue}{rgb}{0.05,0.25,0.65}
\definecolor{darkgreen}{RGB}{20,140,10}
\definecolor{lightgray}{rgb}{0.9,0.9,0.9}
\definecolor{darkorange}{RGB}{200,100,5}
\definecolor{darkyellow}{rgb}{.91,.91,0}
\definecolor{lightolive}{RGB}{225, 220, 185}

\let\originalsslash\sslash
\renewcommand{\sslash}{\mathord{\originalsslash}}

\makeatletter
\newcommand{\cpt}{\mathpalette\cpt@inner\relax}
\newcommand{\cpt@inner}[2]{%
  \scalebox{0.5}[0.9]{$#1\cup$}
  #1\{\infty\}
}
\makeatother

\makeatletter
\newcommand{\plus}{\mathpalette\sqcpt@inner\relax}
\newcommand{\sqcpt@inner}[2]{%
  \scalebox{0.5}[0.9]{$#1\sqcup$}
  #1\{\infty\}
}
\makeatother

\tikzset{
  snake left/.style={
    rounded corners,
    to path={
      let \p1 = (\tikztostart.east),
          \p2 = (\tikztotarget.west),
          \p3 = ($(\p1)!0.5!(\p2)$),
          \n1 = {8pt} 
      in
      (\p1)
      -- (\x1 + \n1, \y1)
      -- (\x1 + \n1, \y3)
      -- (\x2 - \n1, \y3) \tikztonodes
      -- (\x2 - \n1, \y2)
      -- (\p2)
    }
  }
}

\tikzset{
  uphordown/.style={
    rounded corners,
    to path={
      let \p1 = (\tikztostart.north),
          \p2 = (\tikztotarget.north),
          \n1 = {max(\y1,\y2) + 8pt}
      in
      (\p1)
      -- (\x1, \n1)
      -- (\x2, \n1) \tikztonodes 
      -- (\p2)
    }
  }
}

\tikzset{
  downhorup/.style={
    rounded corners,
    to path={
      let \p1 = (\tikztostart.south),
          \p2 = (\tikztotarget.south),
          \n1 = {min(\y1,\y2) - 8pt}
      in
      (\p1)
      -- (\x1, \n1)
      -- (\x2, \n1) \tikztonodes 
      -- (\p2)
    }
  }
}

\tikzset{
  rightvertleft/.style={
    rounded corners,
    to path={
      let \p1 = (\tikztostart.east),
          \p2 = (\tikztotarget.east),
          \n1 = {max(\x1,\x2) + 8pt}
      in
      (\p1)
      -- (\n1, \y1)
      -- (\n1, \y2) \tikztonodes 
      -- (\p2)
    }
  }
}

\tikzset{
  leftvertright/.style={
    rounded corners,
    to path={
      let \p1 = (\tikztostart.west),
          \p2 = (\tikztotarget.west),
          \n1 = {min(\x1,\x2) - 8pt}
      in
      (\p1)
      -- (\n1, \y1)
      -- (\n1, \y2) \tikztonodes 
      -- (\p2)
    }
  }
}

\newcommand{\inlinetikzcd}[1]{\begin{tikzcd}[sep=small, ampersand replacement=\&]#1\end{tikzcd}}

\renewcommand{\setminus}{-}

\newcommand{\defneq}{\mathrel{\equiv}}

\newcommand{\HilbertSpace}{%
  \mathcal{H}%
}

\colorlet{IntBraneColor}{darkblue}
\colorlet{FracBraneColor}{darkorange}

\begin{document}

\setlength{\abovedisplayskip}{3.5pt}
\setlength{\belowdisplayskip}{3.5pt}
\setlength{\abovedisplayshortskip}{-5pt}
\setlength{\belowdisplayshortskip}{3pt}



\title
{Topological Quantum Gates on 
\\
FCIs with Bounded Quantum Geometry}

\author[1,2]
{\NameEmail
  {Hisham Sati}
  {hsati@nyu.edu}
}

\author[1]
{\NameEmail
  {Urs Schreiber}
  {us13@nyu.edu}
}

\affil[1]{Mathematics Program and Center for Quantum and Topological Systems (CQTS), \newline New York University Abu Dhabi, UAE}
\affil[2]{The Courant Institute of Mathematical Sciences, \newline New York University, New York, USA}

\maketitle

\begin{abstract}
Two-band 2D fractional Chern insulators (FCIs) have recently been shown, in principle, to admit unusual anyonic order, associated with Bloch Hopfions localized in momentum space. Here we show that bounding the integrated quantum metric along one primitive Brillouin-zone cycle may refine this topological order to $|C|$ distinct braiding observables, cyclically permuted by lattice translations, for Chern numbers $|C| > 1$. The cyclic structure entails mildly non-abelian braiding phases and otherwise matches the spectral flow exhibited by the FCI's hybrid Wannier states, with the metric bound equivalently a bound on the states' real-space spread. 

We give a classification of quantum gates in principle implementable on this topological order; among them we find the non-Clifford T-gate and the non-Clifford entangling CCZ-gate. This points to novel pathways to fault-tolerant quantum processes within the known orbital structure of FCIs, potentially informing experimental searches towards hardware for future topological quantum computers.
\end{abstract}

\keywords{
  anyons,
  topological quantum gates,
  holonomic quantum computing,
  topological order,
  fractional Chern insulators,
  Hopfions
}

\MSCcodes{
  Primary:
  81V27, 
  81P65, 
  81V70; 
 Secondary:
  55P35, 
  81Q70 
}

\funding{
  by \textit{Tamkeen UAE} under the \textit{Abu Dhabi Research Institute Grant} \texttt{CG008}
}

\newpage
\tableofcontents

\section
{Motivation}

\paragraph
{The Prospect of FCI Topological Order}
Fractional Chern insulators (FCIs, \cite{Tang2011,Sun2011,Neupert2011}, cf. the review in \cite{Neupert2015}) are potentially topologically ordered \cite{Wen1995} crystalline topological quantum materials \cite{GuptaSaxena18}, which realize a momentum-space analog of fractional quantum Hall liquids
(FQH, \cite{Stormer99}). The latter provide firmly established experimental
realizations of abelian \emph{anyons} \cite{Goldin2023}, whose characteristic fractional statistics have been observed by independent groups in $\mathrm{GaAs}$ heterostructures and in graphene%
~\cite{Bartolomei2020,Nakamura2020,Werkmeister2025}.

While analogous (dual, cf. \cite[Fig. 3]{SS25-FQAH}), FCIs, as recently realized in twisted $\mathrm{MoTe}_2$ and
rhombohedral graphene~\cite{Park2023,Lu2024},
do not require a strong external magnetic field and have the prospect of operating at less extreme temperatures \cite{Tang2011} than FQH liquids. These properties make FCIs a potentially realistic platform for
future quantum hardware for topological quantum computing devices (\cite{Kitaev2003,FreedmanKitaevLarsenWang2003,Nayak2008}, cf. \cite{Urton2023})---provided suitable anyonic excitations in FCIs can be identified and controlled. This remains to be seen experimentally,%
\footnote{
Recent optical experiments provide encouraging evidence for fractional charges in FCI states~\cite{Ki2026}, but their braiding statistics, if any, remain to be observed.
}
but crucially it already remains to be properly understood theoretically.

\paragraph
{New Route to Nonabelian Topological Order}
In particular, universal topological quantum gates require universal \emph{nonabelian} anyons (\cite{FreedmanLarsenWang2002}, cf. \parencites{Bravyi2006}{MySS2024}) whose actual experimental realization has remained open, even theoretically.%
\footnote{
The much discussed FQH $\nu = \sfrac{5}{2}$ state, while theoretically expected to feature non-universal Ising anyons, appears unstable and remains experimentally elusive \parencites{Pan1999}[p. 1]{TrungYang2026}.
Note that recent claims of nonabelian anyon experiments (such as \cite{Lo2026universal}) concern their quantum \emph{simulation} on noisy non-topological hardware, cf. \parencites[p. 1470]{Xu2024}{GammonSmith2026}.}

Recently, by revisiting theoretical foundations (briefly recalled in \cref{FCIHopfionAnyons} below), we have observed \parencites{SS25-FQAH,SS25-Crys} that 2D 2-band FCIs may potentially exhibit a non-traditional but theoretically well-motivated form of topological order associated with anyonic solitons localized in the crystal's momentum space, namely in the form of \emph{Hopfion} structures \cite{WilczekZee1983} in the Bloch Hamiltonian, cf. \cref{TheHopfionGroup} below. 

This derivation \cite{SS25-FQAH}, however, relied on the idealized assumption that the system's topological monodromy is evaluated within the entire space of mathematically admissible Bloch Hamiltonians, while realistic systems will be subject to bounds.

Here we consider the more physically realistic situation where a bound is imposed on the system's \emph{quantum metric} (recalled in \cref{FCIQuantumGeometry} below) integrated along one primitive Brillouin-zone cycle.
Our main result (in \cref{Results}) is that this bound actually enhances the Hopfionic topological order to potentially enable the realization of topological quantum gates such as the T-gate. 

\begin{remark}
  The question of identifying microscopic (lattice) models potentially realizing the parameter monodromy, which we classify here on more abstract grounds, is beyond the scope of this letter. We suggest reading our abstract classification result as mathematically establishing novel theoretical possibilities for anyonic topological order which may be worth the attention of model builders and experimentalists.
\end{remark}

In particular, we see that this enhanced topological order goes along with a cyclic layering of FQH phases which is of just the same form as expected when translating hybrid FCI Wannier states into Laughlin states for FQH systems. We first recall this layering in \cref{FCIQuantumGeometry}.

\section
{Methods}
\label{Methods}

\subsection
{FCI Quantum Geometry}
\label
{FCIQuantumGeometry}

We briefly recall how \emph{quantum geometry} (cf. \cite{Kolodrubetz2017}) controls the localization of band-projected $x$-position eigenstates in 2D crystalline materials, and how these same states exhibit an effective cyclic layer of QH states, an observation due to \parencites[\S I]{Qi2011}[\S I]{BarkeshliQi2012}.

\paragraph
{Quantum Metric}

The \emph{quantum metric} \cite{ProvostVallee1980} of a smooth family of pure quantum
states (hence of rays in a given Hilbert space $\HilbertSpace$)
\begin{equation}
  \begin{tikzcd}[
    sep=0pt
  ]
    \Sigma
    \ar[
      rr,
      "{  
      }"
    ]
    &&
    P{
      \HilbertSpace
    }
    \\
    \sigma 
      &\mapsto&
    \bracket|{\psi_\sigma}\rangle
    \bracket\langle{\psi_\sigma}\vert
    \mathrlap{
      \,,
      \;\;\;
      \bracketmid\langle{
        \psi_\sigma
      }{
        \psi_\sigma
      }\rangle
      = 
      1
      \,,
    }
  \end{tikzcd}
\end{equation}
is the pullback (to $\Sigma$)
of the Fubini--Study metric on the complex projective space $P\HilbertSpace$, measuring the infinitesimal distinguishability of neighboring states.

Specialized to crystalline band theory, the parameters are crystal momenta
$\mathbf{k} \in T^d$ varying in the \emph{Brillouin torus} (cf. \parencites[\S 2.1]{Thiang2025} and \cref{TheBrillouinTorus}), and the states are the cell-periodic Bloch
eigenstates $\bracket\vert{u(\mathbf{k})}\rangle$ of an isolated band. For a locally chosen normalized
eigenvector $\bracket|{u(\mathbf{k})}\rangle$, with
projector $P(\mathbf{k}) = \bracket|{u(\mathbf{k})}\rangle\bracket\langle {u(\mathbf{k})}|$,
the quantum metric is
\begin{equation}
  \label{TheQuantumMetric}
  g_{ij}(\mathbf{k})
  =
  \operatorname{Re}
  \bigl\langle
    \partial_{k_i}u(\mathbf{k})
    \big|
    \bigl(1-P(\mathbf{k})\bigr)
    \big|
    \partial_{k_j}u(\mathbf{k})
  \bigr\rangle .
\end{equation}
This expression measures the variation of the band eigenspace across the Brillouin zone.

\paragraph
{Hybrid Wannier states}
This momentum-space geometry determines a real-space localization scale: For a single isolated band in two dimensions,
the maximally localized \emph{hybrid Wannier states} form a family
\begin{equation}
  \label{HybridWannierStates}
  \bracket|{
    W_n(k_y)
  }\rangle,
  \qquad 
  n \in \mathbb{Z}
  \mathrlap{\,,}
\end{equation}
where $n$ labels lattice translations along $x$.
These states are localized along $x$ while retaining transverse
crystal momentum $k_y$, and are eigenstates of the
band-projected position operator
$\widehat P\hat x\widehat P$, where $\widehat P$ projects onto
the full band~\cite[Sec.~II.H]{MarzariEtAl2012}.

Applying the one-dimensional spread identity of \cite[Sec.~III and App.~C]{MarzariVanderbilt1997} at each fixed $k_y$ identifies the physical position variance
of these hybrid Wannier states with the $x$-average of the
quantum metric:
\begin{equation}
  \label{HybridWannierSpreadQuantumMetric}
  \mathrm{Var}_{W_n(k_y)}(\hat x)
  =
  \frac{a_x}{2\pi}
  \int_0^{2\pi/a_x} g_{xx}(k_x,k_y)
  \,
  \mathrm{d}k_x
  =:
  \overline{g}_{xx}(k_y),
\end{equation}
where $a_{x/y}$ denote the lattice periods, cf. \cref{TheBrillouinTorus}.
Consequently:
\begin{standout}
A uniform bound on the $x$-averaged quantum metric is equivalent to a uniform bound on the physical $x$-position variance of the hybrid Wannier eigenstates of projected position:
\begin{equation}
\label
{MetricBoundIsWannierSpreadBound}
  \overline{g}_{xx}(k_y)\leq\ell^2
      \;\;\;\text{for all }k_y
  \;\;\;\;\;
    \Leftrightarrow
  \;\;\;\;\;
   \mathrm{Var}_{W_n(k_y)}(\hat x)\leq\ell^2
     \;\;\;\text{for all }k_y,n .
\end{equation}
\end{standout}
Imposing the same bound throughout the allowed Hamiltonian
deformations ensures that these eigenstates retain a
root-mean-square width along $x$ of at most $\ell$.

\begin{SCfigure}[2.8][htb]
\adjustbox{
  rndfbox=4pt,
  raise=-.5cm
}{
\begin{tikzpicture}[
  >={
    Computer Modern Rightarrow[
      length=4pt, width=4pt
    ]
  }
]

\draw[
  fill=lightgray,
]
  (0,0) rectangle (2,2);

\node at (1,.8) {%
  \clap{\smash{$T^2$}}%
};

\draw[
  line width=3,
  color=darkorange
] 
  (0,0) --
  (0,2);
\draw[
  ->,
  line width=1
]
  (0,1) --
  (0,1.01);

\draw[
  line width=3,
  color=darkorange
] 
  (2,0) --
  (2,2);
\draw[
  ->,
  line width=1
]
  (2,1) --
  (2,1.01);

\draw[
  line width=3,
  color=darkgreen
] 
  (0,2) --
  (2,2);
\draw[
  ->,
  line width=1
]
  (1,2) --
  (1.01,2);

\draw[
  line width=3,
  color=darkgreen
] 
  (0,0) --
  (2,0);
\draw[
  ->,
  line width=1
]
  (1,0) --
  (1.01,0);

\draw[
  fill=black
] 
(0,0) circle (.1);
\draw[
  fill=black
] 
(2,0) circle (.1);
\draw[
  fill=black
] 
(0,2) circle (.1);
\draw[
  fill=black
] 
(2,2) circle (.1);

\draw[
  -Latex
]
  (-.5,-.3) 
  to node[
    below,
    scale=.8
  ]{ $k_x$ }
  (2.5, -.3);
\draw[
  -Latex
]
  (-.3,-.5) 
   to node[
     left,
     scale=.8
   ]
     {$k_y$}
  (-.3, 2.5);

\node[scale=.8] at (0,-.4) 
  {$\overset{\vert}{0}$};
\node[scale=.8] at (2,-.51) 
  {$\overset{\vert}{\tfrac{2\pi}{a_x}}$};
\node[scale=.8] at (-.4,0)
  {$0 -$};
\node[scale=.8] at (-.48,1.98)
  {$\tfrac{2\pi}{a_y} -$};

\node[scale=.7] at (.9,.2) {$W_a$};
\node[scale=.7] at (.3,.9) {$W_b$};

\end{tikzpicture}
}
\caption{\label{TheBrillouinTorus}%
The Brillouin zone of an effectively two-dimensional crystalline system is a torus $T^2$, depicted here by a fundamental domain with opposite edges identified.
The momentum coordinates $k_x$ and $k_y$ have reciprocal periods
$2\pi/a_x$ and $2\pi/a_y$, respectively; the displayed aspect ratio is schematic.
Below, in \cref{Results}, we constrain the $k_x$-averaged quantum metric
$\overline{g}_{xx}(k_y)$ \cref{TheQuantumMetric} uniformly in $k_y$, selecting the horizontal cycle as the ``small'' momentum direction. The associated momentum 1-cycles are $W_a$ for this bounded direction and $W_b$ for the unrestricted vertical direction, cf. \cref{FigureTorusAlgebra}.}
\end{SCfigure}
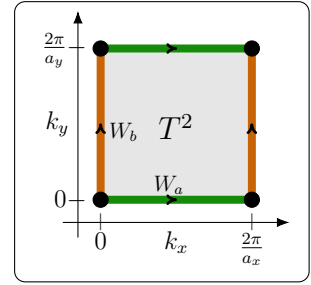

\paragraph
{Cyclic QH Layers of FCI states}
A striking insight of Qi and Barkeshli--Qi is that the hybrid
Wannier states $\bracket\vert{ W_n(k_y) }\rangle$ \cref{HybridWannierStates} of a single Chern band exhibit an effective
multilayer quantum Hall structure~\cite{Qi2011,BarkeshliQi2012} in that:
\begin{enumerate}
\item 
the integer index $n$ in \cref{HybridWannierStates} labels successive lattice translates:
\begin{equation}
  x_n(k_y)
  :=
  \bracketmidtwo\langle 
    {W_n(k_y)}
    {\hat x}
    {W_n(k_y)}
  \rangle
  =
  n \, a_x
    +
  \overline{x}(k_y),
\end{equation}
where 
$\overline{x}(k_y)$ is a common offset determined by the Berry phase
around the $k_x$-cycle;

\item
for a band of Chern number $C \neq 0$, traversing the
transverse Brillouin-zone cycle shifts each center by
$C$ lattice cells, so that \cite[Eqs.~(2)--(3)]{Qi2011}
\begin{equation}
  \label{HybridWannierSpectralFlow}
  \bracket|{
    W_n\bracket({
      k_y
        +
      2\pi/a_y
    })
  }\rangle
  =
  \bracket|{
    W_{n+C}(k_y)
  }\rangle
  \mathrlap{\,.}
\end{equation}

Continuation in $k_y$ therefore joins the orbitals into
$|C|$ families, distinguished by the residue
\begin{equation}
  n\bmod C \in \mathbb{Z}_{/C}
  \mathrlap{\,.}
\end{equation}
These are the effective \emph{layers},
cyclically permuted upon translation $T_x$ by one lattice cell (cf. \cref{LayersSchematics}):
\begin{equation}
\label{TranslationAndCyclicPermutation}
  T_x
  \bracket|{W_n(k_y)}\rangle
  =
  \bracket|{W_{n+1}(k_y)}\rangle,
  \qquad
  n \mapsto  n+1
  \mod C
  \quad\text{in }\mathbb{Z}_{/C}.
\end{equation}
Thus continuation around the transverse momentum cycle
connects orbitals within each layer, while lattice translation
cycles between the layers~\cite[Sec.~II]{BarkeshliQi2012}.

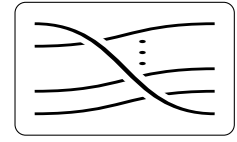
\begin{SCfigure}[3.3][htb]
\caption{\label{LayersSchematics}%
  The eigenstates of projected $x$-position operators in an FCI at Chern number $C \neq 0$ exhibit an orbital representation consisting of $|C|$ Quantum Hall layers that are cyclically permuted under lattice translation in the $x$-direction \cref{TranslationAndCyclicPermutation}.
}
\adjustbox{
  scale={.7}{1.4},
  rotate=90,
  rndfbox=4pt
}{
\begin{tikzpicture}[
  scale=0.85,
]

  \draw[line width=1.5]
    (.5,0) .. controls
    (.5,1) and
    (0,1) ..
    (0,2);
  \draw[line width=1.5]
    (1,0) .. controls
    (1,1) and
    (.5,1) ..
    (.5,2);

  \node at 
    (1.4,.8) {\textbf{\dots}};

  \draw[line width=1.5]
    (2,0) .. controls
    (2,1) and
    (1.5,1) ..
    (1.5,2);

  \draw[
    line width=6,
    white
  ]
    (0,0) .. controls
    (0,1) and
    (2,1) ..
    (2,2);
  \draw[line width=1.5]
    (0,0) .. controls
    (0,1) and
    (2,1) ..
    (2,2);

\end{tikzpicture}
}
\end{SCfigure}

\item 
With this, one finds \cite[below (5)]{Qi2011} that these Wannier states in a layer
\[
  n=i+Cq,
  \qquad
  i\in\{0,\ldots,|C|-1\},
  \qquad q\in\mathbb{Z},
\]
correspond to Quantum Hall Landau orbitals at transverse momentum 
\begin{equation}
  K 
    := 
  k_y 
    + 
  q 
  \tfrac{2\pi}{a_y}
  \mathrlap{\,.}
\end{equation}

Multilayer FQH wavefunctions can then be transferred to the Chern band by replacing their Landau-level orbitals
with these Wannier orbitals while retaining their occupation-number coefficients.

\end{enumerate}

This layered organization is encoded in the spectral flow of projected position; the metric bound in
\cref{HybridWannierSpreadQuantumMetric} controls the physical
widths of the orbitals that realize it.
(We will see this structure reappear from novel topological considerations in \cref{LLsmFundamentalGroups} below.)

\subsection
{FCI Hopfion Anyons}
\label
{FCIHopfionAnyons}

But this raises a key question that may not have been addressed before:
\begin{standout}
  If the bounded quantum geometry of FCIs is effectively that of cyclic FQH layers, how may the corresponding layers of FQH anyons arise?
\end{standout}

The answer we consider here crucially draws on the recent observation \cite{SS25-FQAH,SS25-Crys} that 2-band FCIs in principle admit momentum-space localized anyons in the guise of Hopfion solitons in the Bloch Hamiltonian. We briefly recall how this comes about.

\paragraph
{Topological Order from Parameter Monodromy}

Anyonic statistics may be understood as a special case of the more general
notion of \emph{topological adiabatic holonomy/monodromy}
(\parencites{ZanardiRasetti1999}{PachosZanardiRasetti2000}{PachosZanardi2001}, cf.
\parencites[\S 3]{MySS2024}[\S II.B]{SS25-Crys}).
Consider a family of gapped many-body ground-state spaces
$\mathcal H_p$, depending on parameters $p\in\mathcal P$.
When the topological part of adiabatic transport depends only
on the homotopy class of the parameter path, these spaces
form a \emph{local system}, equivalently a flat unitary
vector bundle over $\mathcal P$.
Its holonomy on each connected component is a representation of the \emph{fundamental group} of continuous deformation paths of the system's parameters in topological phase $C \in \pi_0\bracket({ \mathcal{P} })$:
\begin{equation}
\label{GroundStateMonodromy}
\inlinetikzcd{
  \widehat{(\text{--})}
  :
  \pi_1\bracket({\mathcal P,C})
  \ar[r]
  \&
  \mathrm U(\mathcal H_{C})
\mathrlap{\,.}
}
\end{equation}

In the familiar case where the parameters are configurations of $N$ quasiparticles on a surface $\Sigma$,
\begin{equation}
  \mathcal{P}_{\mathrm{trd}}
  \defneq
  \bracketmid\{{
    (x_1, \cdots, x_N)
    \in
    \Sigma^N
  }{
    \forall_{i \neq j} :  x_i \neq x_j
  }\}
  \,,
  \;\;
  \Rightarrow
  \;\;
  \pi_1\bracket({
    \mathcal{P}_{\mathrm{trd}}
  })
  \simeq
  \mathrm{PBr}\bracket({\Sigma})
  \,,
\end{equation}
the fundamental group \cref{GroundStateMonodromy} is the surface's (pure) \emph{braid group}, whose unitary representations are the familiar anyon representations.
But parameter monodromy describes topological order more generally.

In particular, a crystalline system's external parameters are instead hopping amplitudes or crystal potentials, encoded by single-electron \emph{Bloch Hamiltonians} $H$, and topologically by their \emph{spectral flattenings} $\widehat{H}$, which are continuous maps $\inlinetikzcd{T^d \ar[r] \& \mathrm{Gr}_v\bracket({\mathbb{C}^N})}$:
\begin{enumerate}
\item
from the crystal's \emph{Brillouin torus} $T^d$ of crystal quasi-momenta,

\item 
to the \emph{Grassmannian space} $\mathrm{Gr}_v\bracket({\mathbb C^N})$ of $v \in \mathbb{N}$ valence states among $N$ available bands, 
\end{enumerate}
and hence form the $d$-fold free loop space $L^d$ of the Grassmannian:
\begin{equation}
  \label
  {BlochHamiltonianParameterSpace}
  \mathcal{P}_{\mathrm{crs}}
  \defneq
  \left\{
  \begin{tikzcd}[
    sep=-2pt
  ]
    T^d 
    \ar[rr]
    &&
    \mathrm{Gr}_{v}\bracket({
      \mathbb{C}^N
    })
    \\
    \mathbf{k}
    &\longmapsto&
    \widehat H(\mathbf{k})
  \end{tikzcd}
  \right\}
  =:
  \operatorname{Map}
  \bracket({
    T^d,
    \operatorname{Gr}_v\bracket({
      \mathbb C^N
    })
  })
  \simeq
  L^d 
    \operatorname{Gr}_v\bracket({
      \mathbb C^N
    })
  \mathrlap{\,.}
\end{equation}
Hence for pure crystalline systems without further real-space defects, the available parameter monodromy is that of deformations of (spectrally flattened) Bloch Hamiltonians:
\begin{equation}
\label{EquivariantMonodromyInIntroduction}
\pi_1\bracket({
  \mathcal{P}_{\mathrm{crs}},
  \bracket[{\widehat H}]
})
 \,=\,
\left\{
\begin{tikzcd}[
  column sep=30pt]
  T^d
  \ar[
    rr,
    bend left=30,
    phantom,
    "{}"{name=s,swap}
  ]
  \ar[
    rr,
    bend right=30,
    phantom,
    shift left=3pt,
    "{}"{name=t}
  ]
  \ar[
    from=t,to=s,
    shorten=-4pt,
    Rightarrow,
    shift right=4pt,
    bend right=60,
    gray,
    scale=.5
  ]
  \ar[
    rr,
    bend left=30,
    "{\widehat H}"{description}
  ]
  \ar[
    rr,
    bend right=30,
    gray,
    shift left=3pt,
    crossing over
  ]
  \ar[
    from=s,to=t,
    shorten=-1pt,
    Rightarrow,
    shift right=4pt,
    bend right=60,
    crossing over
  ]
  &&
  \mathrm{Gr}_v
  \mathrlap{%
    \bracket({%
      \mathbb{C}^N%
    })%
  }
  \phantom{%
    \bracket({%
      \mathbb{C}^N%
    })%
  }
\end{tikzcd}
\right\}
\mathrlap{\,.}
\end{equation}

For example, in a two-band system,
\begin{equation}
\label{TheBlochSphere}
  \operatorname{Gr}_1(\mathbb C^2)
    \defneq
  \mathbb CP^1
    \simeq 
  S^2
\end{equation} 
is the Bloch sphere (cf. \parencites[p. 112]{Sergeev2023}). The observation of \parencites{SS25-FQAH}{SS25-Crys} is that in this case the Bloch monodromy \cref{EquivariantMonodromyInIntroduction} is precisely that expected for abelian anyons on a torus (here the momentum-space Brillouin torus): 
\begin{proposition}[{\parencites[Thm. 1]{SS25-FQAH}[Thm. 3.5]{KSS26-HigherDimAnyons}, following \parencites{Hansen1974}{LarmoreThomas1980}{Kallel2001}}]

\label[proposition]{LLFundamentalGroups}
For 2-dimensional crystalline 2-band systems, with Bloch parameter space given by \cref{BlochHamiltonianParameterSpace,TheBlochSphere}, 
\begin{equation}
\label
{2Dim2BandBlochHamiltonianSpace}
  \mathcal{P}_{\mathrm{crs}} 
    = 
  L L S^2
\end{equation}
\begin{enumerate}
\item
the topological phases are labeled by the integers \textup{(Chern numbers, as usual)}:
\begin{equation}
  \pi_0 
  \bracket({
    L L S^2
  })
  \;\simeq\;
  \mathbb{Z}
  \mathrlap{\,,}
\end{equation}

\item
\label
{LLFundamentalGroupGenerators}
the topological monodromy groups are \emph{integer Heisenberg groups}:%
\footnote{
  For a pair of group elements $g_1, g_2 \in G$, their \emph{group commutator} is $[g_1, g_2] 
    := 
  g_1 \, g_2 \, g_1^{-1} \, g_2^{-1} 
  \in G$.
}
\begin{equation}
\label{TheLLFundamentalGroupGenerators}
  \pi_1 \bracket({
    L L S^2
    ,\,
    C
  })
  \simeq
  \bracketmid\langle{
  \begin{array}{c}
    W_a,
    \\
    W_b,
    \\
    \zeta
  \end{array}
  }{
    \begin{aligned}
      \bracket[{W_a, W_b}]
      & = \zeta^2
      \\
      \bracket[{W_a, \zeta}] & = 1
      \\
      \bracket[{W_b, \zeta}] & = 1
      \\
      \zeta^{2C} & = 1
    \end{aligned}
  \,}\rangle \,.
\end{equation}
\end{enumerate}
\end{proposition}
Here:
\begin{enumerate}
\item 
$W_a,W_b$ describe transport around a basis of torus 1-cycles, cf. \cref{FigureTorusAlgebra};
\item
$\zeta$ generates the phase picked up by an elementary braiding process taking place locally, in a chart $\inlinetikzcd{\mathbb{R}^2 \ar[r, hook] \& T^2 }$ of the torus, mathematically reflected \parencites{SS25-AbelianAnyons}{SS25-WilsonLoops} as a Bloch \emph{Hopfion} \parencites{WilczekZee1983}:
\begin{equation}
\label
{TheHopfionGroup}
  \pi_1
  \,
  \mathrm{Map}^\ast
  \bracket({
    \mathbb{R}^{2}_{\cup \{\infty\}}
    ,
    S^2
  })
  \simeq
  \pi_3\bracket({
    S^2
  })
  \simeq
  \mathbb{Z}
  \langle\zeta\rangle
  \mathrlap{\,.}
\end{equation}
\end{enumerate}
Hence in a superselection sector, being an irreducible unitary representation $\widehat{(\text{--})}$ \cref{GroundStateMonodromy} of the monodromy group \cref{TheLLFundamentalGroupGenerators} on the system's Hilbert space of gapped ground states,
we have
\begin{equation}
\label
{IrrepOfHeisenbergGroup}
 \begin{aligned}
 \widehat{\zeta} 
   & =
   e^{i\theta}\mathbf 1
 \\
 \widehat W_a \widehat W_b
 & =
 e^{2i\theta}
 \widehat W_b
 \widehat W_a,
 \end{aligned}
\end{equation}
for $e^{i \theta}$ a root of unity (assuming a finite-dimensional space of gapped ground states, cf. \cite[Lem. 3.32]{SS25-FQH}).
Remarkably, this is exactly the torus Wilson-loop algebra expected for abelian anyons in FQH systems
(\parencites[(4.9)]{WenNiu1990}[(4.14)]{IengoLechner1992}, reviewed in \parencites[(4.21)]{Fradkin2013}[(5.28)]{Tong2016}, cf. Fig. \ref{FigureTorusAlgebra}) but here appearing over the Brillouin torus of crystal momenta.

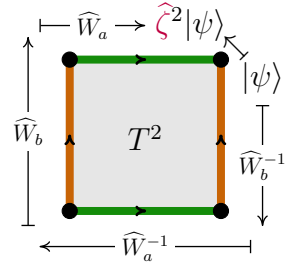
\begin{SCfigure}[1.7][htb]
\adjustbox{
}{
\begin{tikzpicture}[
  >={
    Computer Modern Rightarrow[
      length=4pt, width=4pt
    ]
  }
]

\draw[
  fill=lightgray,
]
  (0,0) rectangle (2,2);

\node at (1,.8) {%
  \clap{\smash{$T^2$}}%
};

\draw[
  line width=3,
  color=darkorange
] 
  (0,0) --
  (0,2);
\draw[
  ->,
  line width=1
]
  (0,1) --
  (0,1.01);

\draw[
  line width=3,
  color=darkorange
] 
  (2,0) --
  (2,2);
\draw[
  ->,
  line width=1
]
  (2,1) --
  (2,1.01);

\draw[
  line width=3,
  color=darkgreen
] 
  (0,2) --
  (2,2);
\draw[
  ->,
  line width=1
]
  (1,2) --
  (1.01,2);

\draw[
  line width=3,
  color=darkgreen
] 
  (0,0) --
  (2,0);
\draw[
  ->,
  line width=1
]
  (1,0) --
  (1.01,0);

\draw[
  fill=black
] 
(0,0) circle (.1);
\draw[
  fill=black
] 
(2,0) circle (.1);
\draw[
  fill=black
] 
(0,2) circle (.1);
\draw[
  fill=black
] 
(2,2) circle (.1);

\node at (2-.4,2+.5) {
  $\mathcolor{purple}{\widehat{\zeta}}^2\vert \psi \rangle$
};

\node[
  rotate=+135  
] at (2+.2, 2+.2) {%
  \clap{$\mapsto$}%
};

\node at (2.55,2-.25) {
  $\vert \psi \rangle$
};

\draw[
  line width=.5,
  |->,
]
  (2.55, 1.4) --
  node[
    xshift=0pt,
    scale=.75
  ] {\colorbox{white}{$
    \widehat W_b
      ^{-1}
  $}}
  (2.55,-.2);

\draw[
  line width=.5,
  |->,
]
  (2.4, -.47) --
  node[
    yshift=0pt,
    scale=.75
  ]{\colorbox{white}{$
    \widehat W_a^{-1}
  $}}
  (-.4, -.47);

\draw[
  line width=.5,
  |->,
]
  (-.55, -.2) --
  node[
    xshift=0pt,
    scale=.75
  ] {\colorbox{white}{$
    \widehat{W}_b
  $}}
  (-.55,2.3);

\draw[
  line width=.5,
  |->,
]
  (-.4, 2.45) --
  node[
    scale=.75,
  ]{\colorbox{white}{$
    \widehat{W}_a
  $}}
  (1, 2.45);
  
\end{tikzpicture}
}
\caption{\label{FigureTorusAlgebra}%
The relations \cref{TheLLFundamentalGroupGenerators,IrrepOfHeisenbergGroup} say that an adiabatic parameter loop of Bloch Hopfions around a basis of Brillouin torus cycles (cf. \cref{TheBrillouinTorus}) induces transformation of the topological states by multiplication with the square of a complex phase
$\zeta$, a root of unity, to be identified \cite[Prop. 3.21]{SS25-FQH} with an anyon braiding phase.}
\end{SCfigure}

\paragraph
{General Covariance and Modular Data}

In more detail, topological states, subject to the laws of topological field theory, are \emph{generally covariant} \cite{Witten1988}, in that they must see diffeomorphisms of space as gauge equivalences. Here this means \parencites[\S II.C]{SS25-Crys}[\S 2.2]{SS25-FQH} that the actual parameter space for crystalline topological phases is not quite the space $\mathcal{P}_{\mathrm{crs}}$ of Bloch Hamiltonians \cref{BlochHamiltonianParameterSpace} itself, but its homotopy quotient $(\text{--})\sslash(\text{--})$ (Borel construction) by the precomposition action with the (orientation-preserving) diffeomorphism group $\mathrm{Diff}^+\bracket({\text{--}})$:
\begin{equation}
\label
{GenerallyCovariantizedParameterSpace}
  \mathcal{P}_{\mathrm{crs}}^{\mathrm{cov}}
  :=
  \mathrm{Map}\bracket({
    T^d, 
    \mathrm{Gr}_v\bracket({
      \mathbb{C}^N
    })
  })
  \sslash
  \mathrm{Diff}^+\bracket({
    T^d
  })
  \mathrlap{\,.}
\end{equation}
Recall that the connected components of the diffeomorphism group of the 2-torus
\begin{equation}
\label{TheModularGroup}
  \pi_0 
  \, 
  \mathrm{Diff}^+\bracket({T^2})
  \simeq
  \mathrm{SL}\bracket({
    2,\mathbb{Z}
  })
  \simeq
  \langle
    S, T
  \rangle
\end{equation}
form the classical \emph{modular group}, generated by the familiar matrices (cf. \parencites[Ex. 2.19]{SS25-FQH})
\begin{equation}
\label
{ModularGenerators}
  S 
    = 
  \left[
  \begin{matrix}
    0 & 1 
    \\
    -1 & 0
  \end{matrix}
  \right]
  \,,
  \;\;
  \;
  T 
    =
  \left[
  \begin{matrix}
    1 & 1 
    \\
    0 & 1
  \end{matrix}
  \right]
  \mathrlap{\,.}
\end{equation}

\begin{proposition}[{\cite[Prop. 3.20]{SS25-FQH}}]
\label[proposition]{ModularActionOnHeisenberg}
  For even $C \in 2 \mathbb{Z}$, the isomorphism in \cref{TheLLFundamentalGroupGenerators} may be chosen so that the diffeomorphism action \cref{GenerallyCovariantizedParameterSpace} on $L L S^2 = \mathrm{Map}\bracket({T^2, S^2})$ gives the semidirect product 
  \begin{equation}
  \label
  {ModularSemidirectProductWithHeisenberg}
    \pi_1\bracket({
      \mathcal{P}^{\mathrm{cov}}_{\mathrm{crs}}
      ,C
    })
    \defneq
    \pi_1\bracket({
      \bracket({
        L L S^2
      })\sslash
      \mathrm{Diff}^+\bracket({T^2})
      , 
      C
    })
    \simeq
  \bracketmid\langle{
  \begin{array}{c}
    W_a,
    \\
    W_b,
    \\
    \zeta
  \end{array}
  }{
    \begin{aligned}
      \bracket[{W_a, W_b}]
      & = \zeta^2
      \\
      \bracket[{W_a, \zeta}] & = 1
      \\
      \bracket[{W_b, \zeta}] & = 1
      \\
      \zeta^{2C} & = 1
      \\
      W_a^C & = 1
      \\
      \zeta^{C^2} W_b^C & = 1 
    \end{aligned}
  \,}\rangle
    \rtimes
    \mathrm{SL}\bracket({
      2, \mathbb{Z}
    })
    \mathrlap{\,,}
  \end{equation}
  with the modular group element $ \left(\begin{matrix} p & q \\ r & s \end{matrix}\right) \in \mathrm{SL}\bracket({2,\mathbb{Z}})$ 
  \cref{TheModularGroup} acting on these generators as
  \begin{equation}
    \label
    {TheModularActionOnHeisenberg}
    \begin{tikzcd}[
      row sep=-2pt, column sep=10pt
    ]
      \pi_1\bracket({
        L L S^2, C
      })
      \ar[
        rr,
      ]
      &&
      \pi_1\bracket({
        L L S^2, C
      })
      \\
      W_a \zeta^C 
      &\longmapsto&
      \bracket({
        W_a \zeta^C
      })^p
      \bracket({
        W_b
      })^q
      \zeta^{-pq}
      \\
      W_b &\longmapsto&
      \bracket({
        W_a
        \zeta^C
      })^r
      \bracket({
        W_b
      })^s
      \zeta^{-rs}
      \\
      \zeta 
        &\longmapsto& 
      \zeta
      \mathrlap{\,.}
    \end{tikzcd}
  \end{equation}
\end{proposition}

These results \cref{LLFundamentalGroups,ModularActionOnHeisenberg} predict the possibility of abelian anyons localized in FCI momentum spaces, in the form of Hopf solitons in the Bloch Hamiltonian. 
In fact, they already predict the possibility of nonabelian topological order, reflecting modular transformations $S,T \in \mathrm{SL}\bracket({2,\mathbb{Z}})$ \cref{ModularSemidirectProductWithHeisenberg}, but the anyon braiding phase $\zeta$ itself remains central, hence abelian (by \cref{ModularActionOnHeisenberg} above).

\medskip

However, this derivation is subject to a mathematical idealization, in that the above classifies topological monodromy \cref{EquivariantMonodromyInIntroduction} under arbitrary excursions of system deformations in the mathematical space of Bloch Hamiltonians. Realistic systems will instead explore only a bounded subspace of all theoretically available deformations.  

Our key result in \cref{Results} is that considering a natural quantification of such a bound actually refines and strengthens the prediction of Bloch Hopfion anyons.

\section
{Results}
\label{Results}

For a spectrally flattened Bloch Hamiltonian
\cref{BlochHamiltonianParameterSpace} of a 2-band system \cref{TheBlochSphere},
\begin{equation}
  \widehat H(\mathbf{k})
    =
  \mathbf{n}(\mathbf{k})
    \cdot
  \boldsymbol{\sigma}
  \;\;
  \text{with}
  \;\;
  \mathbf{n}(\mathbf{k}) \in S^2
  \subset \mathbb{R}^3
  \mathrlap{\,,}
\end{equation}
the quantum metric
\cref{TheQuantumMetric} 
is (cf. \cite[(8)]{Piechon2016})
\[
  g_{ij}
  =
  \tfrac{1}{4}\,
  \partial_{k_i}\mathbf{n}\cdot
  \partial_{k_j}\mathbf{n}
  \mathrlap{\,,}
\]
hence its $x$-average \cref{HybridWannierSpreadQuantumMetric} at a given $k_y$ is
\begin{equation}
\label
{AveragedQuantumMetricFor2Bands}
  \overline g_{xx}(k_y)
  =
  \tfrac
    {a_x^2}
    {16\pi^2}
  \int_0^1
    \bracket|{
      \partial_s\mathbf{n}\bracket({
        \tfrac{2\pi s}{a_x}
        ,
        k_y
      })
     }|^2\
     \,
     \mathrm{d}s
  \,
  \mathrlap{\,.}
\end{equation}
For an averaged metric bound 
\begin{equation}
\label
{MetricBound}
  M \in \mathbb{R}_{>  0}
\end{equation}
write
\begin{equation}
\label
{BoundedLoopSpace}
  L_{\leq M} S^2
  :=
  \bracketmid\{{
    \gamma 
  }{
    \tfrac
      {a_x^2}
      {16\pi^2}
    \int_0^1
    \bracket|{
      \partial_s \gamma
     }|^2\
     \,
     \mathrm{d}s
     \;\leq\; 
     M
  }\}
    \subset
  W^{1,2}\bracket({
    S^1, 
    S^2
  })
    \subset
  L S^2
\end{equation}
for the subspace of the free loop space on the absolutely continuous loops for which the Dirichlet functional \cref{AveragedQuantumMetricFor2Bands} is well-defined and bounded by $M$.

\begin{proposition}
  The first critical value of $M$ \cref{MetricBound} at which the homotopy type of $L_{\leq M} S^2$ \cref{BoundedLoopSpace} changes is 
  \begin{equation}
  \label
  {CriticalAverageQuantumMetric}
    M_1 :=
    \bracket({
    \frac
      {a_x}
      {2}
    })^2
    \mathrlap{\,.}
  \end{equation}
\end{proposition}
\begin{proof}
  This follows from classical results of \emph{Morse--Bott theory}, recalled in \cite[\S 3.2]{GSS26-KfromM}.
\end{proof}

\begin{remark}
This bound \cref{CriticalAverageQuantumMetric} says that the root-mean-square width of the corresponding hybrid Wannier state \cref{HybridWannierSpreadQuantumMetric} in the $x$-direction is, by \cref{MetricBoundIsWannierSpreadBound}, at most half the crystal cell's $x$-width. Hence with this bound, the Wannier states are localized just on the scale of a lattice cell, in the $x$-direction.
\end{remark}

\begin{definition}[The configuration space of $x$-bounded Bloch Hamiltonians]
\label[definition]{SmallLoopSpace}
We abbreviate the bounded loop space \cref{BoundedLoopSpace} at the critical bound \cref{CriticalAverageQuantumMetric} of the averaged quantum metric by
\begin{equation}
\label
{TheSmallLoopSpace}
  L_1 S^2 
    :=
  L_{\leq M_1}
  S^2
\end{equation}
and call this the \emph{small loop space} of $S^2$ (following \parencites[\S 3.2]{GSS26-KfromM}).
This means that the corresponding subspace
\begin{equation}
\label
{SpaceOfBoundedBlochHamiltonians}
  L L_1 S^2
  \subset 
  L L S^2
  \simeq
  \mathrm{Map}\bracket({
    T^2,
    S^2
  })
\end{equation}
of the space of 2D 2-band Bloch Hamiltonians \cref{2Dim2BandBlochHamiltonianSpace} is that of those Hamiltonians whose $x$-averaged quantum metric is bounded by \cref{CriticalAverageQuantumMetric}.
\end{definition}

Dependence of topological quantum states on this bounded parameter space may support different topological order, and this is what we will classify now.

First, we have the following simplified description of \cref{TheSmallLoopSpace}, which the impatient reader may regard as a more explicit definition:

\begin{SCfigure}[1.5][htb]
\adjustbox{
  rndfbox=4pt,
  padding=2pt
}{
\includegraphics[width=6cm]
  {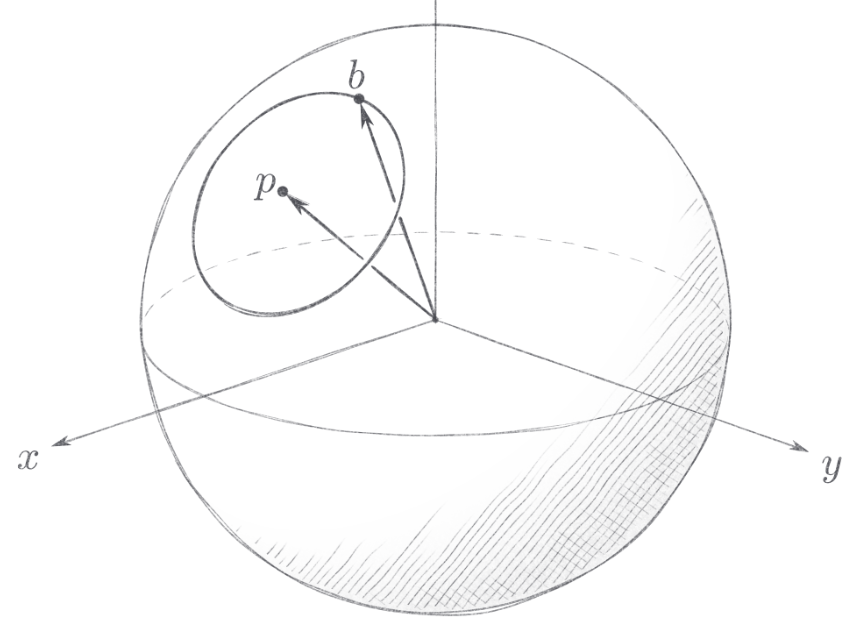}
}
\caption{\label{RoundCircle}
The space $L_1 S^2 \subset L S^2$ \cref{TheSmallLoopSpace} of \emph{small loops} in the (Bloch) sphere 
is homotopy equivalent (by \cref{SmallLoopsAsRoundCircles}) to its further subspace of loops that trace out round circles on $S^2$, once at a constant rate. 
\newline
These in turn are equivalently determined by their base point $b \in S^2$ and by the oriented rotation axis $p \in S^2 \subset \mathbb{R}^3$ that they are encircling clockwise.
\newline
Here constant loops are included, regarded as circles of vanishing radius that as such may be regarded as oriented either by the axis $p=b$ or $p = -b$. 
}
\end{SCfigure}

\begin{lemma}
\label[lemma]
{SmallLoopsAsRoundCircles}
  The space $L_1 S^2$ \cref{TheSmallLoopSpace} of ``small loops'' in $S^2$:
  \begin{enumerate}
  \item
  is homotopy equivalent more specifically to
  the space of loops that traverse (once, at a constant rate) round circles embedded in $S^2$, including degenerate circles corresponding to constant loops,
  \item
  as such may be coordinatized by pairs of points $(b,p) \in S^2 \times S^2$ indicating the loop's basepoint $b$ and the oriented axis $p$ encircled clockwise by their round-circle trajectory, subject only to the relation $\forall_{b \in S^2} : (b,b) \sim (b,-b)$, reflecting that a degenerate circle constant on a point $b$ is oriented either way \textup{(cf. \cref{RoundCircle})}:
  \begin{equation}
  \label
  {RoundCirclesAmongSmallLoops}
    \begin{tikzcd}[
      row sep=-1pt, column sep=small
    ]
      \frac{
        \mbox{$S^2 \times S^2$} 
      }{
        \forall_b : 
        (b,b) \sim (b,-b)
      }
      \ar[
        rr,
        "{ \sim }"
      ]
      &&
      L_1 S^2
      \\
      {[b,p]}
      &\longmapsto&
      \substack{
        \textup{loop $\gamma\!:\!S^1\!\to\!S^2$ starting at $\gamma(0)\!=\!b$} 
        \\
        \textup{going once clockwise at a constant rate}
        \\
        \textup{through the round circle around $p$.}
      }
    \end{tikzcd}
  \end{equation}
  \end{enumerate}
\end{lemma}
\begin{proof}
  This follows by combining classical results in Morse-Bott theory given in \cite[\S 2.3--2.5]{Klingenberg1978} (using Thm. 2.4.10, Cor. 2.4.11 \& Prop. 2.5.3 there), in direct analogy with the description of $L_1 S^4$ that appears as \cite[Lem. 3.4]{GSS26-KfromM} (cf. the proof of Prop. 3.12 there).
\end{proof}

\subsection
{The FQH Layers}
\label
{TheFQHLayers}

Our central observation is the following variant of \cref{LLFundamentalGroups}:
\begin{theorem}[Fundamental groups of $L L_{1} S^2$]
\label[theorem]
{LLsmFundamentalGroups}
The parameter space \cref{SpaceOfBoundedBlochHamiltonians} of $x$-bounded Bloch Hamiltonians has:
\begin{enumerate}
\item
connected components still indexed by the Chern number
$C\in\mathbb{Z}$:
\begin{equation}
  \pi_0
  \bracket({
    L L_{1} S^2
  })
  \simeq
  \mathbb{Z}
  \mathrlap{\,,}
\end{equation}

\item
\label
{LLsmFundamentalGroupGenerators}
fundamental groups
\begin{enumerate}

\item 
for $C = 0$:
\begin{equation}
  \pi_1\bracket({
    L L_{1} S^2
    ,
    0
  })
  \cong
  \bracketmid\langle{
    \begin{array}{c}
      W_a,
      \\
      W_b,
      \\
      \{\zeta_i\}_{i \in \mathbb{Z}}
    \end{array}
    }{
    \begin{aligned}
      [W_a,W_b] &= \zeta_0^2,\\
      [W_a,\zeta_i] &= \zeta_{i+1}\zeta_i^{-1},\\
      [W_b,\zeta_i] &= 1,\\
      [\zeta_i,\zeta_j] &= 1
    \end{aligned}
    \;
    \begin{subarray}{l}
      \textup{for} \,
      i,j\,\in\mathbb{Z}
    \end{subarray}
    }
  \rangle
  \mathrlap{\,,}
\end{equation}

\item 
for $C\neq 0$:
\begin{equation}
\label
{TheLayerdHeisenbergGroupAtNontrivialC}
  \pi_1\bracket({
    L L_{1} S^2
    ,
    C
  })
  \cong
  \bracketmid\langle{\,
    \begin{aligned}
      W_a,
      \\
      W_b,
      \\
      \zeta_0\,,
      \\
      \smash{\vdots}\;,%
      \mathclap{\phantom{\zeta}}
      \\
      \zeta_{
        |C| - 1
      }
    \end{aligned}
    }{
    \begin{aligned}
      [W_a,W_b] &= \zeta_0^2,\\
      [W_a,\zeta_i] &= \zeta_{i+1}\zeta_i^{-1},\\
      [W_b,\zeta_i] &= 1,\\
      [\zeta_i,\zeta_j] &= 1,
      \\
        \zeta_0^2
        \zeta_1^2
          \cdots
        \zeta^2_{
          |C| - 1
        }
        &= 1
    \end{aligned}
    \;
    \begin{subarray}{l}
      \textup{for} \,
      i,j\,\in\mathbb{Z}_{/C}
    \end{subarray}
    }
  \rangle
  \mathrlap{\,.}
\end{equation}

\end{enumerate}
\end{enumerate}
\end{theorem}
\begin{proof}[Proof outline\textup{\footnotemark}]
\footnotetext{
  A \texttt{Lean-4} verification certificate of \cref{LLsmFundamentalGroups} is available at \href{https://ncatlab.org/schreiber/show/IIA+Anyons}{\texttt{ncatlab.org/schreiber/show/IIA+Anyons}}.
}
  With \cref{SmallLoopsAsRoundCircles} we are reduced to computing the fundamental groups of $L Q$ for \cref{RoundCirclesAmongSmallLoops}
  \begin{equation}
    Q := \bracket({ S^2 \times S^2 })/_\sim
    \mathrlap{\,.}
  \end{equation}
  Since $Q$ is the quotient of a connected space, it is connected:
  \begin{equation}
    \pi_0(Q) \simeq \ast
    \mathrlap{\,.}
  \end{equation}
  With the generalized van Kampen theorem (cf. \cite[\S~6.2, Case~2 and Prop.~6.2.2]{Geoghegan2008}),
  one finds
  \begin{equation}
    \pi_1(Q) 
      \simeq
    \mathbb{Z}\langle w_a \rangle
    \mathrlap{\,,}
  \end{equation}
  where we name group generators in angle brackets.
A model for the universal cover of $Q$ consists of an
infinite chain of copies of $S^2\times S^2$, with successive
copies glued along their diagonal and antidiagonal spheres.
Mayer--Vietoris and the Hurewicz theorem then identify
\begin{equation}
\label{pi2Q}
  \pi_2(Q)
    \cong
  \begin{aligned}
    \mathbb{Z}\langle  w_b \rangle
    \mathrlap{\oplus}
  \\
    \textstyle{\bigoplus_{i\in\mathbb Z}}
    \,
    \mathbb{Z}\langle \,z_i\, \rangle\mathrlap{,}
  \end{aligned}
\end{equation}
where under the canonical action
\inlinetikzcd{
  \rho_{(\text{--})}(\text{--})
  :
    \pi_1(Q) \times \pi_2(Q)
    \ar[r]
    \&
    \pi_2(Q)
}
we have
\begin{equation}
\label{waAction}
  \rho_{w_a}
  :
  \left\{
  \begin{tikzcd}[sep=0pt]
    w_b
      & \mapsto &
      w_b 
      &+& 
      2 z_0,
    \\
      z_i
      & \mapsto &
      & &
      z_{i+1}
      \mathrlap{\,.}
  \end{tikzcd}
  \right.
\end{equation}

The inclusion $L Q \subset \mathrm{Map}\bracket({T^2, S^2})$ is clearly a bijection on $\pi_0$, so that the components of $L Q$ are labeled by winding number $C\in\mathbb Z$, as claimed.
In component $C$, the homotopy exact sequence of the
free-loop evaluation fibration has boundary map
$\delta = 1 - \rho_{w_a}^C$, yielding a split exact sequence of this form:
\begin{equation}
\label
{TheSES}
\begin{tikzcd}[row sep=15pt, column sep=25pt]
  1
  \ar[r]
  &
  \frac
   {\mbox{$\pi_2(Q)$}}
   {%
    \rule{0pt}{7pt}%
    ({
      \rho_{w_a}^C - 1
    })
    \pi_2(Q)
  }
  \ar[
    rr,
    hook,
    "{ \iota_C }"
  ]
  &&
  \pi_1\bracket({
    \bracket({
      L Q
    })_C
  })
  \ar[
    rr,
    "{
      \mathrm{ev}_\ast
    }"
  ]
  &&
  \mathbb{Z}\langle w_a \rangle
  \ar[r]
  \ar[
    ll,
    downhorup,
    "{
      W_a \,\mapsfrom\, w_a
    }"{description}
  ]
  &
  1
  \mathrlap{\,,}
\end{tikzcd}
\end{equation}
with
\begin{equation}
  W_a 
    \,\iota_C[\text{--}]\, 
  W_a^{-1}
  =
  \iota_C
  \bracket[{\rho_{w_a}(\text{--})}]
  \mathrlap{\,.}
\end{equation}

This implies that our group is the semidirect product
\begin{equation}
\label
{TheSemidirectProductStructure}
  \begin{tikzcd}[
    row sep=-2pt, column sep=10pt
  ]
    \pi_1\bracket({
      \bracket({ LQ })_C
    })
    \ar[
      rr,
      <->,
      "{ \sim }"
    ]
    &&
    \frac
     {\mbox{$\pi_2(Q)$}}
    {%
      \rule{0pt}{7pt}%
      ({
        \rho_{w_a}^C - 1
      })
      \pi_2(Q)
    }
    \rtimes_{\rho}
    \mathbb{Z}\langle w_a \rangle
    \\
    W_a &\longleftrightarrow& w_a
    \\
    W_b &\longleftrightarrow& w_b
    \\
    \zeta_i &\longleftrightarrow& z_i
    \mathrlap{\,,}
  \end{tikzcd}
\end{equation}
where \cref{waAction} gives
\begin{equation}
  \begin{aligned}
    W_a \, W_b \,  W_a^{-1}
    & =
    \zeta_0^2 W_b
    \\
    W_a \, \zeta_i \,  W_a^{-1}    
    & =
    \zeta_{i + 1}
  \end{aligned}
\end{equation}
while \cref{pi2Q} gives
\begin{equation}
  \begin{aligned}
    W_b \, \zeta_i \,  W_b^{-1}    
    & =
    \zeta_{i}
    .
  \end{aligned}
\end{equation}
With this, the quotient on the left of \cref{TheSES} enforces, for $C > 0$,
\begin{equation}
  \begin{tikzcd}[sep=0pt]
  &
  W_a^C 
  \bracket({ W_b }) W_a^{-C}
  & = &
  W_b
  \\
  \Leftrightarrow
  &
  W_a^{C - 1} \bracket({
    \zeta_0^2 \, W_b 
  }) W_a^{-C + 1}
  & = &
  W_b
  \\
  \Leftrightarrow
  &
  W_a^{C - 2} \bracket({
    \zeta_1^2 \zeta_0^2 \, W_b 
  }) W_a^{-C + 2}
  & = &
  W_b
  \\
  \Leftrightarrow
  &
  & \vdots
  \\
  \Leftrightarrow
  &
  \zeta_{C - 1}^2 \cdots \zeta_1^2 \zeta_0^2 \, W_b 
  &=&
  W_b
  \\
  \Leftrightarrow
  &
  \zeta_{C - 1}^2 \cdots \zeta_1^2 \zeta_0^2
  &=&
  1
  \mathrlap{\,,}
  \end{tikzcd}
\end{equation}
and
\begin{equation}
  \begin{tikzcd}[sep=0pt]
  &
  W_a^C 
  \bracket({ \zeta_i }) W_a^{-C}
  & = &
  \zeta_i
  \\
  \Leftrightarrow
  &
  W_a^{C - 1} \bracket({
    \zeta_{i + 1} 
  }) W_a^{-C + 1}
  & = &
  \zeta_i
  \\
  \Leftrightarrow
  &
  W_a^{C - 2} \bracket({
    \zeta_{i + 2}
  }) W_a^{-C + 2}
  & = &
  \zeta_i
  \\
  \Leftrightarrow
  &
  & \vdots
  \\
  \Leftrightarrow
  &
  \zeta_{i + C}
  &=&
  \zeta_i
  \mathrlap{\,.}
  \end{tikzcd}
\end{equation}
Equivalent relations are implied for $C < 0$, while for $C = 0$ the quotient imposes no relation.

In summary, these are equivalently the claimed generators and relations.
\end{proof}

\begin{example}
\label
{LayeredHeisenbergAtChernNumber2}
  For Chern number $\vert C\vert = 2$, the group \cref{TheLayerdHeisenbergGroupAtNontrivialC} is equivalently presented as
  \begin{equation}
    \begin{aligned}
    &
    \pi_1\bracket({
      L L_1 S^2
      ,
      2
    })
    \simeq
    \\
    &
    \bracketmid\langle{\,  
    \begin{aligned}
      & W_a,
      \\
      & W_b,
      \\
      & \,\zeta_0,
      \\
      & \,\zeta_1
    \end{aligned}
    }{
    \begin{alignedat}{3}
      W_a \, W_a \, W_a^{-1}
      &
      =
      \phantom{\zeta_0^2}
      W_a
      \,,
      &\phantom{--}
      W_b \, W_a \, W_b^{-1}
      &
      =
      \zeta_0^{-2}
      W_a  
      \,,
      \\
      W_a \, W_b \, W_a^{-1}
      &
      =
      \zeta_0^2
      \, 
      W_b
      \,,
      &\phantom{--}
      W_b \, W_b \, W_b^{-1}
      &
      =
      \phantom{\zeta_0^{-2}}
      W_b      
      \,,
      \\
      W_a \; \zeta_0 \; W_a^{-1}
      & =
      \zeta_1
      \,,
      &
      W_b \; \zeta_0 \; W_b^{-1}
      & =
      \zeta_0
      \,,
      \\
      W_a \; \zeta_1 \; W_a^{-1}
      & =
      \zeta_0
      \,,
      &
      W_b \; \zeta_1 \; W_b^{-1}
      & =
      \zeta_1      
      \,,
    \end{alignedat}
    \;\;
    \begin{aligned}
      \phantom{A}
      \\
      \phantom{A}
      \\
      \zeta_0 \zeta_1 
      & = 
      \zeta_1 \zeta_0
      \\
      \zeta_1^2 
      &=
      \zeta_0^{-2}
    \end{aligned}
    \,}\rangle
    \,.
    \end{aligned}
  \end{equation}
\end{example}

\begin{proposition}[Comparison of monodromy as the quantum metric bound is lifted]
\label[proposition]
{thm:forgetting-small-circle-constraint}
\begin{enumerate}
\item
The defining inclusion \cref{SpaceOfBoundedBlochHamiltonians}
\begin{equation}
\label
{TheDefiningInclusion}
  \begin{tikzcd}[column sep=large]
    L L_{1} S^2
      \ar[r, hook, "\iota"]
    &
    L L S^2
  \end{tikzcd}
\end{equation}
induces, for each degree $C \in \mathbb{Z}$, a surjective homomorphism of fundamental groups, which in terms of the generators of 
\textup{\cref{LLFundamentalGroups} 
\cref{LLFundamentalGroupGenerators}}
and
\textup{\cref{LLsmFundamentalGroups} 
\cref{LLsmFundamentalGroupGenerators}}
is:
\begin{equation}
\label
{GroupHomomorphismsInducedByDefiningInclusion}
  \begin{tikzcd}[
    sep=0pt
  ]
    \pi_1 \bracket({
      L L_{1} S^2 , C
    })
      \ar[
        rr, 
        ->>, 
        "{ \iota_\ast } "
      ]
    &&
    \pi_1 \bracket({
      L L S^2 , C
    })
    \\
    W_a &\longmapsto& W_a,
    \\
    W_b &\longmapsto& W_b,
    \\
    \zeta_i &\longmapsto& \zeta
    \mathrlap{\,.}
  \end{tikzcd}
\end{equation}
Equivalently, this is quotienting by the relation
\begin{equation}
  \forall_{i,j}
  :
  \zeta_i \sim \zeta_j
  \mathrlap{\,.}
\end{equation}

\item
The kernels of the homomorphisms \cref{GroupHomomorphismsInducedByDefiningInclusion} are free abelian:
\begin{equation}
  \operatorname{ker}(\iota_\ast)
  \cong
  \begin{cases}
    \bigoplus_{\mathbb{Z} \setminus \{0\}}
    \mathbb{Z}
    &
    \textup{ if }
    C=0,
    \\[6pt]
    \bigoplus_{|C|-1}
    \mathbb{Z}
    & 
    \textup{ if }
    C \neq 0,
  \end{cases}
\end{equation}
with a free abelian basis given by
\begin{equation}
  \begin{cases}
    \big\{
      \zeta_i \zeta_0^{-1}
    \big\}_{i \in \mathbb{Z} \setminus \{0\}}
    & 
    \textup{ if }
    C=0,
    \\[4pt]
    \big\{
      \zeta_i \zeta_0^{-1}
    \big\}_{1 \leq i \leq |C|-1}
    & 
    \textup{ if }
    C \neq 0.
  \end{cases}
\end{equation}
\item
In particular, $\iota_\ast$ is an isomorphism precisely for $|C|=1$.

\end{enumerate}
\end{proposition}

\begin{proposition}
\label[proposition]
{ShearGroupAction}
  Under the restriction to bounded quantum geometry \cref{TheDefiningInclusion}, the modular group $\mathrm{SL}(2,\mathbb{Z})$-action \cref{TheModularActionOnHeisenberg} breaks to the \emph{shear subgroup} which preserves the ``oriented small direction'': 
  \begin{equation}
    \mathbb{Z}
    \simeq
    \bracket\langle{T}\rangle
    \simeq
    \bracketmid\{{
      \left(
      \begin{matrix}
        1 & n 
        \\
        0 & 1
      \end{matrix}
      \right)
    }{
      n \in \mathbb{Z}
    }\}
    \subset
    \mathrm{SL}\bracket({
      2,
      \mathbb{Z}
    })
    \mathrlap{\,,}
  \end{equation}
  whose generator $T \in \mathrm{SL}(2,\mathbb{Z})$ \cref{ModularGenerators} acts on the generators from \cref{LLsmFundamentalGroups} as follows:
  \begin{equation}
    \label
    {ShearActionOnPi1LL1S2}
    \begin{tikzcd}[
      row sep=-2pt, column sep=10pt
    ]
      \pi_1\bracket({
        L L_1 S^2, C
      })
      \ar[
        rr,
      ]
      &&
      \pi_1\bracket({
        L L_1 S^2, C
      })
      \\
      W_a
        &\longmapsto&
      W_a W_b \zeta_{-1}^{-1}
      \\
      W_b 
        &\longmapsto&
      W_b
      \\
      \zeta_i 
        &\longmapsto& 
      \zeta_i
      \mathrlap{\,.}
    \end{tikzcd}
  \end{equation}
\end{proposition}

\subsection
{The Topological Gates}

For a given parameter monodromy, the possible superselection sectors of topological orders are the corresponding irreducible unitary representations, whose operators are the predicted possible topologically protected quantum gates. We discuss the general classification and some interesting examples in our case (\cref{TheFQHLayers}), cf. the comments in \cref{OnTheGatesThatAppear} at the end.

\paragraph
{Classification for bounded quantum metric FCIs}

\begin{proposition}[Irrep classification]
\label[proposition]
{IrrepClassification}
  For $C \neq 0$, the finite-dimensional irreducible unitary representations of $\pi_1\bracket({L L_1 S^2, C})$ \cref{TheLayerdHeisenbergGroupAtNontrivialC} are, up to unitary isomorphism, labeled by:
  \begin{enumerate}
  \item
  orbits $[\chi]$ 
  under the automorphism
  $\rho_{w_a}$ \cref{waAction}
  of unitary characters
  \begin{equation}
  \label
  {TheCharacter}
    \begin{tikzcd}[
      row sep=-2pt, column sep=10pt, 
      /tikz/column 3/.append style={anchor=base west},
    ]
      \bracket\langle{
        W_b, 
        \bracket\{{\zeta_i}\}_{i=1}^{|C|}
      }\rangle
      \ar[
        rr,
        "{
          \chi
        }"
      ]
      &&
      \mathrm{U}(1)
      \\
      W_b 
        &\longmapsto&
      \phi_b
      \\
      \zeta_i
      &\longmapsto&
      \phi_i
        ,
        \;\;\;\;\;\;
        \bracket({
          \phi_1 \cdots \phi_{|C|}
        })^2 = 1
        \,,
    \end{tikzcd}
  \end{equation}
  \item 
  further phases 
  $\lambda \in \mathrm{U}(1)$,
  \end{enumerate}
  and are given on the $d$-dimensional Hilbert space 
  $\HilbertSpace :=
    \bigoplus_{j = 0}^{d-1}
    \mathbb{C}|j\rangle
  $, 
  for 

\begin{equation}
\begin{aligned}
  d
  &:=
  \min\bigl\{
    r\geq 1
    \,\big|\,
    \chi\circ\rho_{w_a}^{\,r}=\chi
  \bigr\}
  \\
  &=
  \min
  \bracketmid\{{
    r\geq 1
  }{
    \forall_i : 
    \phi_{i+r}=\phi_i 
    \,,\;
    \left(\prod_{s=0}^{r-1}\phi_s\right)^2=1
  }\}
\end{aligned}
\end{equation}
  the length of the orbit
  \textup{(a divisor of $|C|$)},
  by:
  \begin{equation}
  \label
  {OperatorsInAnIrrep}
   \begin{tikzcd}[
     ampersand replacement=\&,
     /tikz/column 3/.append style={anchor=base west},
     row sep=0pt, column sep=10pt
   ]
     \pi_1\bracket({
       L L_1 S^2, C
     })
     \ar[
       rr,
       "{
         V_{[\chi],\lambda}
       }"
     ]
     \&\&
     \mathrm{U}\bracket({
       \HilbertSpace
     })
     \\
     W_a  
       \&\longmapsto\&
     \left(
     \widehat{W}_a
       : 
       |j\rangle
       \mapsto 
       \begin{cases}
         {|j+1\rangle} 
         & \textup{ if } 0 \leq j < d-1
         \\
         \lambda |0\rangle 
         & \textup{ if } j = d-1
       \end{cases}
     \right)
    \\
    W_b 
      \&\longmapsto\&
    \left(\,
    \widehat{W}_b
      : 
    {|j\rangle}
      \mapsto
    \phi_b
    \bracket({
      \prod_{r=1}^{j}
      \phi_{-r}^{-2}
    })
    |j\rangle
    \right)
    \\
    \zeta_i
      \&\longmapsto\&
    \left(
      \;
      \widehat{\zeta}_i
      \,
      :
      |j \rangle
      \mapsto
      \phi_{i - j}
      \,
      |j \rangle
    \right)
    \mathrlap{\,.}
   \end{tikzcd}
  \end{equation}
\end{proposition}
\begin{proof}
  By \cref{TheSemidirectProductStructure} 
  our group is a semidirect product of the abelian group on the left of \cref{TheCharacter}. Therefore, Mackey theory implies that its irreps are characterized by orbits of its characters, as claimed here (cf. \parencites[Prop. 11.8]{CurtisReiner1962}[\S 8.2]{Serre1977}).
\end{proof}

\begin{remark}
[Cyclic Monodromy and Effective Quantum Hall Layers]

The irrep classification \cref{IrrepClassification} now allows
us to return to the effective quantum Hall layers of
Qi and Barkeshli--Qi~\cite{Qi2011,BarkeshliQi2012} surveyed in \cref{FCIQuantumGeometry}:
Recall that the Wannier spectral flow
\cref{HybridWannierSpectralFlow} organizes the hybrid Wannier
orbitals into $|C|$ families:
Continuation around the transverse momentum cycle connects
orbitals within each family, while translation by one lattice
spacing in the $x$-direction cyclically permutes the families
\cref{TranslationAndCyclicPermutation}.

Our monodromy representations exhibit just this structure.
The covariance relation
\begin{equation}
  \widehat W_a\,
  \widehat\zeta_i\,
  \widehat W_a^{-1}
  =
  \widehat\zeta_{i+1},
  \qquad
  i
    \in
  \mathbb Z_{|C|},
\end{equation}
is realized in \cref{OperatorsInAnIrrep} by cyclic transport
between character sectors, with
\begin{equation}
  \widehat\zeta_i
  \,
  |j\rangle
  =
  \phi_{i-j}|j\rangle,
  \qquad
  j=0,\ldots,d-1.
\end{equation}
Thus the braiding phase depends on the relative cyclic
position of the observable and the character sector.
\end{remark}

\begin{proposition}[Shear-Equivariant Extension of Irreps]
\label[proposition]{ShearCompatibleIrreps}
\begin{enumerate}
\item
An irrep $V_{[\chi],\lambda}$ \cref{OperatorsInAnIrrep},
\textup{from \cref{IrrepClassification}}, extends to a \textup{(necessarily also irreducible)} unitary representation $V_{[\chi],\lambda,\widehat{T}}$ of the semidirect product with the shear group \textup{(\cref{ShearGroupAction})}
\begin{equation}
  \begin{tikzcd}[
    row sep=12pt, column sep=large
  ]
    \pi_1\bracket({
      L L_1 S^2
      ,
      C
    })
    \ar[
      d,
      hook
    ]
    \ar[
      rr,
      "{
        V_{[\chi],\lambda}
      }"
    ]
    &&
    \mathrm{U}\bracket({
      \HilbertSpace
    })
    \\
    \pi_1\bracket({
      L L_1 S^2
      ,
      C
    })
    \rtimes
    \bracket\langle{
      T
    }\rangle
    \ar[
      urr,
      dashed,
      "{
        V_{
          [\chi],
          \lambda,
          \widehat{T}
        }
      }"{swap, sloped}
    ]
  \end{tikzcd}
\end{equation}
if and only if
\begin{equation}
\label{ShearExtensionDeterminantCondition}
  \det\bracket({
    \widehat W_b
  })
    =
  \det\bracket({
    \widehat\zeta_0
  })
  \mathrlap{\,,}
\end{equation}
which, in terms of the character parameters \cref{TheCharacter}, is equivalent to
\begin{equation}
\label{ShearExtensionCharacterCondition}
  \phi_b^{\,d}
  =
  \textstyle{\prod_{r=1}^{d}}
  \phi_{-r}^{\,2(d-r)+1},
\end{equation}
where the subscripts on $\phi$ are understood modulo $C$.

\item
Whenever the conditions \cref{ShearExtensionDeterminantCondition,ShearExtensionCharacterCondition} hold, 
the unitary equivalence classes of extensions form a $\mathrm{U}(1)$-torsor: given one extension $V_{[\chi],\lambda, \widehat{T}}$, any other is uniquely obtained by multiplying the operator $\widehat{T}$ representing $T$ by an element in $\mathrm{U}(1)$.
\end{enumerate}
\end{proposition}
\begin{proof}
\begin{enumerate}
\item
By \cref{ShearActionOnPi1LL1S2}, an extension amounts to a choice of a unitary operator
$\widehat{T}$ such that
\begin{equation}
\label{ShearExtensionIntertwiningRelations}
\begin{aligned}
  \widehat{T}
  \,
  \widehat W_a
  \,
  \widehat{T}^{-1}
    &=
  \widehat W_a
  \,
  \widehat W_b
  \,
  \widehat\zeta_{-1}^{-1},
  \\
  \widehat{T}
  \,
  \widehat W_b
  \,
  \widehat{T}^{-1}
    &=
  \widehat W_b,
  \\
  \widehat{T}\widehat\zeta_i\widehat{T}^{-1}
    &=\widehat\zeta_i
    \qquad\text{for all }i.
\end{aligned}
\end{equation}
Taking determinants in the first relation gives
\[
  \det\bracket({
    \widehat W_b
  })
    =
  \det({
    \widehat\zeta_{-1}
  })
  \mathrlap{\,.}
\]
Since $\widehat\zeta_{-1}$ and $\widehat\zeta_0$
are conjugate by $\widehat W_a$, this proves necessity
of \cref{ShearExtensionDeterminantCondition}.

Conversely, suppose that
\cref{ShearExtensionDeterminantCondition} holds.
In the character basis of \cref{OperatorsInAnIrrep}, write
\[
  \widehat W_b
  |j\rangle
    =
  b_j|j\rangle,
  \qquad
  b_j
    :=
  \phi_b
  \textstyle{\prod_{r=1}^{j}}
  \phi_{-r}^{-2},
  \qquad
  0\leq j <  d.
\]
Now define unit complex numbers $u_j$ recursively by
\begin{equation}
\label{ShearImplementerRecursion}
  u_0=1,
  \qquad
  u_{j+1}
    =u_j\,b_j\,\phi_{-1-j}^{-1},
  \qquad
  0\leq j<d.
\end{equation}
Their product around the cyclic basis satisfies
\[
  \frac{u_d}{u_0}
  =
  \prod_{j=0}^{d-1}
     b_j\phi_{-1-j}^{-1}
  =
  \frac{
    \det\bracket({
      \widehat W_b
    })
  }{
    \det\bracket({
      \widehat\zeta_{-1}
    })}
  =1.
\]
Thus $u_d=u_0$, and the recursion closes consistently.
Set
\begin{equation}
\label{ShearImplementingOperator}
\begin{tikzcd}[sep=0pt]
  \HilbertSpace
    \ar[rr,"{\widehat{T}}"]
  &&
  \HilbertSpace
  \\
  {|j\rangle}
    &\longmapsto&
  u_j|j\rangle .
\end{tikzcd}
\end{equation}
This operator is unitary and commutes with
$\widehat W_b$ and every $\widehat\zeta_i$,
since these operators are diagonal in the same basis.
The recursion \cref{ShearImplementerRecursion} gives
the first relation in
\cref{ShearExtensionIntertwiningRelations}.
This also holds at the wrap from $|d-1\rangle$ to
$|0\rangle$: the factor $\lambda$ in $\widehat W_a$
occurs on both sides and cancels.
Hence $\lambda$ imposes no additional condition.

Since the shear subgroup $\langle  T \rangle \simeq \mathbb{Z}$ is infinite cyclic, there is no
further relation to impose on $\widehat{T}$.
The extension is therefore given by
\[
  V_{[\chi],\lambda,\widehat{T}}\bracket({g,T^n})
    :=
  V_{[\chi],\lambda}(g)
  \widehat{T}^{n}
  \mathrlap{\,.}
\]

Expanding the product in the recursion \cref{ShearImplementerRecursion} gives
\[
  u_d
  =
  \phi_b^{\,d}
  \,
  \textstyle{\prod_{r=1}^{d}}
    \phi_{-r}^{-\left(2(d-r)+1\right)}.
\]
Consequently $u_d=1$ is precisely
\cref{ShearExtensionCharacterCondition}.

\item
If $\widehat{T}'$ is another representation of $T$, 
then $\widehat{T}^{-1}\widehat{T}'$ commutes
with every $V_{[\chi],\lambda}(g)$.
By the assumption that $V_{[\chi], \lambda}$ is irreducible, Schur's lemma implies that $\widehat{T}^{-1}\widehat{T}'$ operates by multiplication with a complex number, and unitarity then implies that this complex number is in $\mathrm{U}(1)$. 
\qedhere
\end{enumerate}
\end{proof}

\newpage 
\paragraph
{Examples of Realizable Quantum Gates}
We highlight some familiar quantum gates (cf. \parencites{BarencoEtAk1995}[\S 4.5]{NielsenChuang2000}) that appear in the irreps classified by \cref{IrrepClassification}, insofar as their topological protection is noteworthy (cf. \cref{OnTheGatesThatAppear}). All of the following examples are chosen such that they are shear-equivariant according to \cref{ShearCompatibleIrreps}.

\subparagraph
{I. Diagonal gates}

The key effect on possible quantum gates of passing from the unbounded parameter monodromy \cref{TheLLFundamentalGroupGenerators} to its bounded counterpart \cref{TheLayerdHeisenbergGroupAtNontrivialC}
is that the $C$th power $\zeta^C$ of a single braiding phase---whose constraint
$\bracket({\zeta^{C}})^2 = 1$ quantizes that braiding phase to a root of unity---is relaxed to a product $\zeta_1 \cdots \zeta_{|C|}$ of a tuple of phases that are individually unconstrained apart from their global condition $\bracket({\zeta_1 \cdots \zeta_{|C|}})^2 = 1$ \cref{TheCharacter}.

As a result, all gates represented by diagonal matrices in the computational basis may be realized via \cref{OperatorsInAnIrrep}, possibly with an ancillary space in order to satisfy that global constraint in \cref{TheCharacter}. We highlight some noteworthy examples and record the Chern numbers $C$ needed to realize them.

\begin{example}[T-Gate Realization]
\label[example]{TGateRealization}
Consider $C=2$ as in \cref{LayeredHeisenbergAtChernNumber2}. 
If we take the parameter phases \cref{TheCharacter} to be:
\begin{equation}
\label
{PhaseChoiceForTGate}
  \phi_b
    =
  \phi_1
    =
  e^{\mathrm{i}\pi/8},
  \qquad
  \phi_2
    =
  e^{-\mathrm{i}\pi/8},
  \qquad
  \lambda 
    = 
  1
  \mathrlap{\,,}
\end{equation}
then \cref{IrrepClassification} demonstrates an irrep (topological order) of dimension $d=2$, given by:
\begin{equation}
\label
{TheQuantumGates}
\begin{tikzcd}[
  ampersand replacement=\&,
  /tikz/column 3/.append style={anchor=base west},
  row sep=-1pt, column sep=small 
]
  \pi_1\bracket({
    L L_1 S^2
    ,
    2
  })
    \ar[
      rr,
      "{ V_{\chi,1} }"
    ]
  \&\&
  \mathrm{U}\bracket({
    \mathbb{C}^2
  })
  \\
  W_a
    \&\longmapsto\&
  \widehat W_a
  =
  \begin{pmatrix}
    0 & 1 \\
    1 & 0
  \end{pmatrix}
  \\
  W_b,\;\zeta_1
    \&\longmapsto\&
  \widehat W_b=\widehat\zeta_1
  =
  \begin{pmatrix}
    e^{\mathrm{i}\pi/8} & 0 \\
    0 & e^{-\mathrm{i}\pi/8}
  \end{pmatrix}
  \\
  \zeta_2
    =
  \zeta_0
    \&\longmapsto\&
  \widehat\zeta_2
    =
  \widehat\zeta_0
    =
  \widehat\zeta_1^{-1}
  \mathrlap{\,.}
\end{tikzcd}
\end{equation}
In terms of quantum gates we see that:
\begin{enumerate}
\item 
  $\widehat W_a
    =
    X
    :=
    \begin{pmatrix}
      0 & 1 
      \\
      1 & 0
    \end{pmatrix}
  $ 
  is the \emph{Pauli $X$-gate}, 

\item 
\label
{TheTGateAppears}
$
  \widehat\zeta_0
  =
  e^{-\mathrm{i}\pi/8}\,T
  :=
  e^{-\mathrm{i}\pi/8}
  \begin{pmatrix}
    1 & 0 \\
    0 & e^{\mathrm{i}\pi/4}
  \end{pmatrix}
$
is (up to a global phase) the 
\emph{T-gate} (cf.
 \parencites
   [p. xxx]{NielsenChuang2000}).
\end{enumerate}
\end{example}

\begin{example}[Controlled Z-Gate Realization]
\label[example]{ControlledPhaseGateRealizations}
For $N \in \mathbb{N}_{\geq 2}$ with $C = 2^N$,
consider the following character phases in \cref{TheCharacter}:
\begin{equation}
  \begin{aligned}
  \phi_b 
    & = 
    e^{\mathrm{i}\pi/C},
  \\
  \phi_i
  &=
  \begin{cases}
    -1 & \text{ if } 
    i = 1 \pmod {2^N}
    \\
    +1 & \text{ otherwise } 
  \end{cases}
  \\
  \lambda
  & = 1\,.
  \end{aligned}
\end{equation}
In this case \cref{IrrepClassification} demonstrates an irrep given by:
\begin{equation}
\begin{tikzcd}[
  ampersand replacement=\&,
  /tikz/column 3/.append style={anchor=base west},
  row sep=0pt, column sep=small 
]
  \pi_1\bracket({
    L L_1 S^2
    ,
    C
   })
    \ar[rr,"{V_{\chi,1}}"]
  \&\&
  \mathrm{U}\bracket({
    \mathbb{C}^{2^N}
  })
  \\
  W_a
    \&\longmapsto\&
  \left(
    \widehat W_a
      :
    |j\rangle
      \longmapsto
    |j\!+\!1\rangle
  \right)
  \\
  W_b
    \&\longmapsto\&
  \left(
    \widehat W_b
    :
    |j\rangle
      \longmapsto
    e^{\mathrm{i}\pi/C}
    |j\rangle
  \right)
  \\
  \zeta_i
    \&\longmapsto\&
  \left(
    \;
    \widehat \zeta_i
    \,
    :
    |j\rangle
      \longmapsto
    \begin{cases}
      - |j\rangle 
      &\text{ if }
      i - j = 1 \pmod {2^N}
      \\
      + |j\rangle
      & \text{ otherwise }
    \end{cases}
  \right)
  \mathrlap{\,.}
\end{tikzcd}
\end{equation}
Identifying the Hilbert space with that of $N$ qubits, by binary expansion:
\begin{equation}
\label
{TheBinaryQBitEncoding}
\begin{tikzcd}[
  sep=0pt,
  /tikz/column 3/.append style={anchor=base west},
]
  \bracket({
    \mathbb C^2
  })^{\otimes N}
    \ar[
      rr,
      "\simeq"
    ]
    &&
  \mathbb{C}^{2^N}
  \\
  \lvert b_1\rangle
    \otimes\cdots\otimes
  \lvert b_N\rangle 
    &\longmapsto&
  \lvert j\rangle
  \,,\;\;
  j
    =
  \sum_{r=1}^{N} 
  2^{N-r}
  b_r,
  \quad b_r\in\{0,1\}
  \mathrlap{\,,}
\end{tikzcd}
\end{equation}
and observing that the binary expansion of $-1 = 2^N -1 \pmod {2^N}$ has $\forall_r : b_r = 1$,
we see that
\begin{equation}
\label
{GeneralControlledZGate}
  \widehat \zeta_0
  :
  |b_1,\ldots,b_N\rangle
  \longmapsto
  (-1)^{b_1\cdots b_N}
  |b_1,\ldots,b_N\rangle
\end{equation}
is the \emph{$(N-1)$-controlled $Z$ gate} $\mathrm{C}^{N-1}\mathrm{Z}$ (cf. 
\parencites
  [(1)]{GuhneEtA2014}
  {TakagiYoderChuang2017}
  [\S 4.1]{BeverlandEtAl2020}).
In particular:
\begin{enumerate}
\item
For $C = 2^2$, this \cref{GeneralControlledZGate} is the two-qubit \emph{controlled-Z gate}
\begin{equation}
\label{TheCZGate}
  \widehat\zeta_0
  =
  \mathrm{CZ}
  =
  \operatorname{diag}(1,1,1,-1).
\end{equation}
This is an entangling gate (cf. \parencites{BremnerEtAl2002}), since
\begin{equation}
  \mathrm{CZ}
  \bracket({
  \bracket({
    |0\rangle
    +
    |1\rangle
  })
  \otimes
  \bracket({
    |0\rangle
    +
    |1\rangle
  })
  })
  =
  |00\rangle
    +
  |01\rangle
    +
  |10\rangle
    -
  |11\rangle
\end{equation}
is an entangled state.

\item
\label
{CaseOfTheCCZGate}
For $C = 2^3$, 
this \cref{GeneralControlledZGate}
is the three-qubit
\emph{controlled-controlled Z-gate}
\begin{equation}
\label{TheCCZGate}
  \widehat\zeta_0
  =
  \mathrm{CCZ}
  =
  \operatorname{diag}(1,1,1,1,1,1,1,-1).
\end{equation}
This gate is in addition non-Clifford (cf. \parencites{CampbellTerhalVuillot2017}): conjugation sends
the Pauli operator $X\otimes I\otimes I$ to
\begin{equation}
  \mathrm{CCZ}
  (X\otimes I\otimes I)
  \mathrm{CCZ}^{-1}
  =
  X\otimes\mathrm{CZ}
  \mathrlap{\,,}
\end{equation}
which is not a Pauli operator.
\end{enumerate}
\end{example}
\begin{example}[Metaplectic Qutrit-Gate Realization]
  Similarly to \cref{ControlledPhaseGateRealizations} but for $C = 3$ we may choose
  \begin{equation}
    \phi_1 = -1
    \,,\;\;
    \phi_2
      =
    \phi_3
      =
    \lambda
      = 
    1
    \,,\;\;
    \phi_b
    =
    e^{\mathrm{i}\pi/C}
  \end{equation}
  and obtain the \emph{metaplectic qutrit gate} (cf. \cite{GlaudellEtAl2022})
  \begin{equation}
    \widehat \zeta_0
    =
    \mathrm{diag}(1,1,-1)
    \mathrlap{\,.}
  \end{equation}
\end{example}

\begin{example}[Controlled Phase-Gate Realizations]
\label[example]
{ControlledPhaseGateRealization}
In a joint generalization of \cref{TGateRealization,ControlledPhaseGateRealizations}, consider $C = 4$ with  
\begin{equation}
  \phi_1 
  =
  \varphi 
  \,
  e^{\mathrm{i}\pi/4}
  \,,\;\;
  \phi_2
    = 
  \phi_3
    =
  \phi_4 
    =
  \varphi
  \,,
  \;\;\;
  \phi_b 
    =
  \varphi
  \,,\;\;
  \lambda
    =
  1
  \mathrlap{\,,}
\end{equation}
for 
$
\varphi
  :=
  e^{ - \mathrm{i}\pi /16}
$.
This yields 
the \emph{controlled T-gate} 
\begin{equation}
  \widehat \zeta_0
  =
  \varphi
  \,
  \mathrm{diag}\bracket({
    1,
    1,
    1,
    e^{\mathrm{i} \pi/4}
  })
\end{equation}
and hence the \emph{controlled S-gate} (cf. \cite{HaahHastings2018})
\begin{equation}
  \widehat \zeta_0^2
  =
  \varphi^2
  \,
  \mathrm{diag}\bracket({
    1,
    1,
    1,
    \mathrm{i}
  })
\end{equation}
in addition to the controlled Z-gate $\propto \widehat \zeta_0^4$ from \cref{TheCZGate},
up to global phases.
\end{example}
\begin{example}[$\sqrt{\mathrm{T}}$-Gate Realization]
In the same $C = 4$ irrep of \cref{ControlledPhaseGateRealization},
the square root of the T-gate (cf. \cite[\S A.2]{BeverlandEtAl2020}),
\begin{equation}
  \sqrt{\mathrm{T}}
  :=
  \mathrm{diag}\bracket({
    1,
    e^{\mathrm{i}\pi/8}
  })
  \mathrlap{\,,}
\end{equation}
may also be realized, with an ancillary space and up to a global phase, as
\begin{equation}
  \widehat W_b
  \,
  \widehat \zeta_0
  \,
  \widehat \zeta_3^{-1}
  =
  e^{- \mathrm{i}\pi/16}
  \mathrm{id}_2 
    \otimes
  \sqrt{\mathrm{T}}
  \mathrlap{\,.}
\end{equation}
\end{example}

\subparagraph
{II. Non-diagonal gates}

The principal non-diagonal gate realized by the irreps of \cref{IrrepClassification} is the \emph{quantum incrementer}, represented by $\widehat W_a$ in \cref{TheLayerdHeisenbergGroupAtNontrivialC} whenever $\lambda = 1$. Since the split braiding phases $\zeta_i$ do not commute with $W_a$, we obtain interesting operator conjugations. 

\begin{example}[Realizing Magic State Stabilizers]
\label[example]{RealizingMagicStateProjectors}
 In the situation of \cref{ControlledPhaseGateRealizations} \cref{CaseOfTheCCZGate}, observe that in the binary qubit encoding
 \cref{TheBinaryQBitEncoding} we have 
 $
   \widehat W_a^4
   \bracket|{
     b_1, b_2, b_3
   }\rangle
   =
   \bracket|{
     b_1 + 1, b_2, b_3
   }\rangle   
 $
 and hence
 \begin{equation}
   \widehat W_a^4 
   =
   \mathrm{X}_1
   :=
   \mathrm{X} 
     \otimes 
   \mathrm{id}
     \otimes
   \mathrm{id}
   \in
   \mathrm{U}\bracket({
     \mathbb{C}^2 
       \otimes
     \mathbb{C}^2 
       \otimes
     \mathbb{C}^2 
   })\,.
 \end{equation}
 Therefore, with \cref{TheCCZGate}, we may realize the operator
 \begin{equation}
   \widehat \zeta_0
   \,
   \widehat W_a^4
   \,
   \widehat \zeta_0^{-1}
   =
   \mathrm{CCZ}
   \,
   \mathrm{X}_1
   \,
   \mathrm{CCZ}^{-1}
   =
   \mathrm{X}
   \otimes
   \mathrm{CZ}
   \mathrlap{\,.}
 \end{equation}
 This quantum gate and its cousins are used for the verification of CCZ magic states (cf. \parencites[\S 5.2]{HaahHastings2018}). 
\end{example}

\section
{Conclusion}

We set out to look for new theoretical pathways towards possibly realizing topological quantum gates on crystalline topological insulators, based on the previously possibly underappreciated observation of \emph{holonomic quantum computing}: The theoretically achievable
\footnote{
  Throughout, we are concerned here with the theoretically achievable quantum orders, leaving the question of their lattice or experimental realization to be addressed elsewhere.
}
topological orders are the \emph{parameter monodromies} of gapped ground-state Hilbert spaces classified in their superselection sectors by (finite-dimensional unitary) irreps of the fundamental groups of the gapped quantum system's external parameter space. Here we highlighted that, for crystalline insulators, the external parameters are the crystal's potential/binding parameters, which are encoded in the single-electron Bloch Hamiltonians. Hence:
\begin{standout}
  Possible topological quantum orders in crystalline insulators are represented by unitary representations of the fundamental groups of the space of maps from their Brillouin torus to the Grassmannian space of valence bands among the accessible bands.
\end{standout}

We specialized this general situation to the archetypical case of 2D 2-band insulators, where it was recently observed that these fundamental groups, when otherwise unconstrained, are integer Heisenberg groups, just those groups whose representations are the expected quantum states of abelian anyons on a torus, traditionally discussed as such for fractional quantum Hall (FQH) systems with their anyonic quasi-particles localized in real space. Here the torus is instead the Brillouin torus of crystal momenta, whence the anyons predicted are localized in momentum space, with their braids being Hopf solitons in paths of Bloch Hamiltonian deformations. 
\begin{standout}
  Possible topological quantum orders of 2D crystalline insulators with 2 accessible bands and no further constraints are just the momentum space duals of abelian FQH anyons on a torus.
\end{standout}

The new aspect studied here is the case where the deformation classes accessible to the topological insulator system are further subject to a bound on the averaged quantum metric along one basis $a$-cycle (in the $k_x$-direction) of the Brillouin torus. We observed that this bound equivalently localizes the projected $x$-position eigenstates of electrons in fractional Chern insulators to the scale of the crystal lattice. Qi \& Barkeshli had previously observed that such projected position eigenstate bases exhibit the FCI as a cyclic layer of $|C|$ quantum Hall systems. This raises the question:
\begin{standout}
  How does the cyclic layer of quantum Hall states exhibited by projected position eigenbases in FCIs support an analogously layered anyonic topological order?
\end{standout}

Due to the generality of the holonomic parameter monodromy classification, we could address this question by rigorous topological analysis: Theory predicts that this expected layered topological order must, if indeed it exists, be reflected in the irreps of the fundamental groups of the space of 2D 2-band Bloch Hamiltonians subject to the bound on their averaged quantum metric in the $x$-direction. Our main observation is therefore a mathematical theorem, and it turns out to state that
\begin{standout}
  When subjected to the bound on the $x$-averaged quantum metric, the single braiding phase $\widehat \zeta$ of the momentum-space localized Hopfionic Bloch anyons, subject to the root-of-unity constraint $\widehat \zeta^C = \pm I$, splits into a cyclic layer of $|C| \geq 1$ operators $\widehat \zeta_i$ subject jointly to the decomposed constraint $\widehat \zeta_1 \cdots \widehat \zeta_{|C|} = \pm I$, and transforming under lattice translation as expected from Qi-Barkeshli layers. 
\end{standout}

The key points here are:
\begin{enumerate}
\item
The emerging QH-layer structure is suggestive of how this abstractly classified topological order may be identified within the known orbital structure of FCIs,

\item 
the class of topological quantum gates in principle realizable on the quantum metric-bounded topological order is considerably richer than that of ordinary abelian FQH anyons seen without the metric bound, due to:
\begin{standout}
\begin{enumerate}
\item
  the appearance of tuples of largely freely tunable braiding phases $\zeta_i$, 
\item
  their mild non-abelianness, being cyclically permuted by conjugation with $x$-lattice translation operators $W_a$.
\end{enumerate}
\end{standout}

\end{enumerate}

\begin{remark}
\label[remark]
{OnTheGatesThatAppear}
\begin{enumerate}
\item
These two facts imply that the quantum metric-bounded topological order in principle allows the realization of protected quantum gates beyond the notorious restrictions of surface-code based approaches (\cite{CampbellTerhalVuillot2017}), such as T-gates.

\item
As the downside of its higher generality,
the theory developed here does not select specific braiding phases among a continuum of possible solutions to the single constraint $\zeta_1 \cdots \zeta_{|C|} = \pm 1$. The question of how specific choices may be stably realized in lattice models or experiment is not addressed here. 

\item And while mildly non-abelian braiding effects do appear here (cf. \cref{RealizingMagicStateProjectors}), the predicted set of topologically protected quantum gates remains far from universal under our assumptions.
\end{enumerate}
\end{remark}

In conclusion, we have shown, by revisiting theoretical foundations of topological quantum order in general and of crystalline insulators specifically, that:
\begin{standout}
There are interesting previously unrecognized theoretical possibilities for topological quantum gates potentially realizable on 2D 2-band crystalline insulators, compatible with the known cyclically layered quantum Hall orbital structure of fractional Chern insulators.
\end{standout}
While their actual realization in lattice models and experiment is beyond the scope of this letter, these theoretical considerations may help guide searches for such realizations.

\printbibliography

\end{document}